\documentclass[fleqn]{lmcs}
\pdfoutput=1
\usepackage[utf8]{inputenc}

\usepackage{lastpage}
\lmcsdoi{21}{2}{29}
\lmcsheading{}{\pageref{LastPage}}{}{}%
{Mar.~06,~2024}{Jun.~27,~2025}{}

\usepackage{hyperref}
\usepackage{mathtools}
\usepackage{xspace}
\usepackage{amssymb}
\RequirePackage{pifont}
\usepackage{stmaryrd}
\usepackage{tikz}
\usetikzlibrary{decorations.pathreplacing}
\tikzstyle{sd}[6]=[sibling distance=#1mm]
\tikzstyle{ld}[5.5]=[level distance=#1mm]
\tikzstyle{mis}[20]=[minimum size=#1mm]
\tikzstyle{is}[3]=[inner sep=#1pt]
\tikzstyle{inner}=[mis=1.8,is=0,draw,fill=white,shape=circle]
\tikzstyle{leaf}=[mis=1.8,is=0,fill=black,shape=circle]
\tikzstyle{stage}=[scale=1]

\tikzstyle{H}=[level distance=7mm,level 1/.style={sd},
level 2/.style={sd},level 3/.style={sd},rotate=180]
\tikzstyle{brace}=[decoration={brace,raise=1.5mm},decorate,line
width=0.7pt]
\tikzstyle{noedge}=[edge from parent/.style={transparent}]

\newcommand{\btkz}{\begin{tikzpicture}}
\newcommand{\etkz}{\end{tikzpicture}}
\newcommand{\seq}[2][n]{{#2_1},\dots,{#2_{#1}}}
\newcommand{\hh}[1][.3]{\hspace{#1mm}}
\newcommand{\SET}[1]{\{\hh#1\hh\}}
\newcommand{\m}[1]{\mathsf{#1}}
\newcommand{\mc}[1]{\mathcal{#1}}
\newcommand{\mr}[1]{\mathrel{#1}}
\newcommand{\xA}{\mc{A}}

\newcommand{\xC}{\mc{C}}
\newcommand{\xD}{\mc{D}}
\newcommand{\xF}{\mc{F}}

\newcommand{\xH}{\mc{H}}

\newcommand{\xTH}{\mc{T}_{\xH}}
\newcommand{\xTS}[1][S]{\mc{T}_#1(\xF,\xV)}
\newcommand{\xCH}{\xC_{\xH}}

\newcommand{\xR}{\mc{R}}
\newcommand{\xS}{\mc{S}}
\newcommand{\xV}{\mc{V}}
\newcommand{\R}{\rightarrow}
\newcommand{\Rb}[1][]{\R_{#1}}
\newcommand{\Rab}[2][]{\R_{#1}^{#2}}
\newcommand{\MC}[2][0]{\makebox[#1mm]{#2}}
\newcommand{\ML}[2][0]{\makebox[#1mm][l]{#2}}
\newcommand{\Item}{\smallskip\item}

\newcommand{\AC}{\m{AC}}
\newcommand{\ac}{=_\m{AC}}

\newcommand{\NN}{\mathbb{N}}
\newcommand{\OO}{\mathbb{O}}

\newcommand{\gO}{>_\m{O}}

\newcommand{\GO}{~\gO~}

\newcommand{\geO}{\geqslant_\m{O}}

\newcommand{\eO}{=_\m{O}}
\newcommand{\EO}{~\eO~}

\newcommand{\lex}{{\m{lex}}}
\newcommand{\mul}{{\m{mul}}}

\newcommand{\DEC}{\xD\m{ec}}

\newcommand{\si}{\mathbin{|}}
\newcommand{\h}{\m{h}}
\newcommand{\I}{\m{i}}
\newcommand{\hs}[1][.3]{\hspace{#1mm}}
\newcommand{\plus}{\hs+\hs}
\newcommand{\Oplus}{\hs\oplus\hs}
\newcommand{\Otimes}{\hs\otimes\hs}

\newcommand{\Emb}{{\mathcal{E}\m{mb}}}

\newcommand{\EVALA}[2][\alpha]{[#1]_\xA(#2)}
\newcommand{\lab}[2][\alpha]{\m{lab}_#1(#2)}
\newcommand{\xFl}{\xF_\m{lab}}

\newcommand{\xHl}{\xH_\m{lab}}
\newcommand{\xRl}{\xR_\m{lab}}

\newcommand{\RHAC}[1][*]{\Rab[\xH/\AC]{#1}}
\newcommand{\RHnAC}[2][*]{\Rab[\xH_#2/\AC]{#1}}
\newcommand{\RRAC}[1][n]{\Rb[\xR_#1/\AC]}
\newcommand{\nR}[1]{\,\xrightarrow{\,\MC[2]{$\scriptstyle #1$}\,}\,}
\newcommand{\RT}[1]{\m{root}(#1)}
\newcommand{\acmpo}{{\m{acmpo}}}

\newcommand{\ACMPO}[1][>]{#1_{\m{acmpo}}}
\newcommand{\ACMPOm}[1][>]{#1_{\m{acmpo}}^\mul}
\newcommand{\ACm}{=_\AC^\mul}
\newcommand{\superterm}{\trianglerighteqslant}
\newcommand{\prsuperterm}{\rhd}
\newcommand{\prsubterm}{\mr{\vartriangleleft}}
\newcommand{\subterm}{\trianglelefteqslant}
\newcommand{\TF}[2][]{{\triangledown_{\!#1}}(#2)}
\newcommand{\msn}{\m{s}(\m{n})}
\newcommand{\msm}{\m{s}(\m{m})}

\newcommand{\AProVE}{\textsf{AProVE}\xspace}

\newcommand{\muterm}{\textsf{muterm}\xspace}
\newcommand{\TTTT}{\textsf{T\kern-0.15em\raisebox{-0.55ex}T\kern-0.15emT%
\kern-0.15em\raisebox{-0.55ex}2}\xspace}

\newcommand{\exaref}[1]{Example~\ref{exa:#1}}
\newcommand{\thmref}[1]{Theorem~\ref{thm:#1}}
\newcommand{\corref}[1]{Corollary~\ref{cor:#1}}
\newcommand{\lemref}[1]{Lemma~\ref{lem:#1}}
\newcommand{\defref}[1]{Definition~\ref{def:#1}}
\newcommand{\secref}[1]{Section~\ref{sec:#1}}
\newcommand{\figref}[1]{Figure~\ref{fig:#1}}

\title{Hydra Battles and AC Termination}
\thanks{The first author is supported by JSPS KAKENHI Grant Number
JP22K11900. Part of this work was performed when the second author was
employed at the Future Value Creation Research Center of Nagoya
University, Japan.}

\author{Nao Hirokawa\lmcsorcid{0000-0002-8499-0501}}[a]
\address{School of Information Science, JAIST, Japan}
\email{hirokawa@jaist.ac.jp}

\author{Aart Middeldorp\lmcsorcid{0000-0001-7366-8464}}[b]
\address{Department of Computer Science, University of Innsbruck, Austria}
\email{aart.middeldorp@uibk.ac.at}

\keywords{battle of Hercules and Hydra, term rewriting, AC termination}

\begin{document}

\begin{abstract}
We present a new encoding of the Battle of Hercules and Hydra as a
rewrite system with AC symbols. Unlike earlier term rewriting encodings,
it faithfully models any strategy of Hercules to beat Hydra. To prove
the termination of our encoding, we employ type introduction in connection
with many-sorted semantic labeling for AC rewriting and AC-MPO, a new
AC compatible reduction order that can be seen as a much weakened version
of AC-RPO.
\end{abstract}

\maketitle

\section{Introduction}
\label{sec:introduction}

The mythological monster Hydra is a dragon-like creature with multiple
heads. Whenever Hercules in his fight chops off a head, more and more new
heads can grow instead, since the beast gets increasingly angry. Here we
model a Hydra as an unordered tree. If Hercules cuts off a leaf
corresponding to a head, the tree is modified in the following way: If the
cut-off node $h$ has a grandparent $n$, then the branch from $n$ to the
parent of $h$ gets multiplied, where the number of copies 
corresponds to the
number of decapitations so far. Hydra dies if there are no heads left, in
that case Hercules wins. The following sequence shows an example fight:
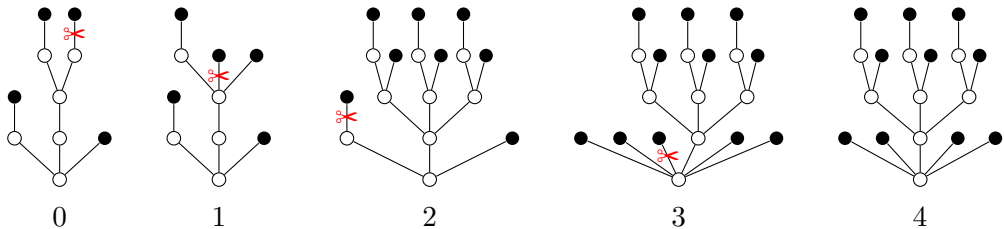
\newcommand{\scissors}{\footnotesize \textcolor{red}{\ding{34}}}
\begin{align*}
&\btkz[ld,level 1/.style={sd},level 2/.style={sd=4},rotate=180]
\node[inner] at (0,0) {}
 child { node[leaf] {} }
 child { node[inner] {}
  child { node[inner] {}
   child { node[inner] (A) {} child { node[leaf] (B) {} } }
   child { node[inner] {} child { node[leaf] {} } } } }
 child { node[inner] {} child { node[leaf] {} } };
\path (A) -- (B) node[midway] () {\scissors};
\node[stage] at (0,.5) {0};
\etkz &
&\btkz[ld,level 1/.style={sd},level 2/.style={sd=5},rotate=180]
\node[inner] at (0,0) {}
 child { node[leaf] {} }
 child { node[inner] {}
  child { node[inner] (A) {}
   child { node[leaf] {} }
   child { node[leaf] (B) {} }
   child { node[inner] {}
    child { node[leaf] {} } } } }
 child { node[inner] {} child { node[leaf] {} } };
\path (A) -- (B) node[midway] () {\scissors};
\node[stage] at (0,.5) {1};
\etkz &
&\btkz[ld,level 1/.style={sd=11},level 2/.style={sd},level 3/.style={sd=3},
 rotate=180]
\node[inner] at (0,0) {}
 child { node[leaf] {} }
 child { node[inner] {}
  child { node[inner] {}
   child { node[leaf] {} }
   child { node[inner] {}
    child { node[leaf] {} } } }
 child { node[inner] {}
  child { node[leaf] {} }
  child { node[inner] {}
   child { node[leaf] {} } } }
 child { node[inner] {}
  child { node[leaf] {} }
  child { node[inner] {}
   child { node[leaf] {} } } } }
 child { node[inner] (A) {}
  child { node[leaf] (B) {} } };
\path (A) -- (B) node[midway] () {\scissors};
\node[stage] at (0,.5) {2};
\etkz &
&\btkz[ld,level 1/.style={sd=5.2},level 2/.style={sd=6.5},
 level 3/.style={sd=2.8},rotate=180]
\node[inner] at (0,0) (A) {}
 child { node[leaf] {} }
 child { node[leaf] {} }
 child { node[inner] {}
  child { node[inner] {}
   child { node[leaf] {} }
   child { node[inner] {}
    child { node[leaf] {} } } }
  child { node[inner] {}
   child { node[leaf] {} }
   child { node[inner] {}
    child { node[leaf] {} } } }
  child { node[inner] {}
   child { node[leaf] {} }
   child { node[inner] {}
    child { node[leaf] {} } } } }
 child { node[leaf] (B) {} }
 child { node[leaf] {} }
 child { node[leaf] {} };
\path (A) -- (B) node[midway,xshift=-.1mm,yshift=.25mm] () {\scissors};
\node[stage] at (0,.5) {3};
\etkz &
\btkz[ld,level 1/.style={sd=5},level 2/.style={sd=6.5},
 level 3/.style={sd=2.8},rotate=180]
\node[inner] at (0,0) {}
 child { node[leaf] {} }
 child { node[leaf] {} }
 child { node[inner] {}
  child { node[inner] {}
   child { node[leaf] {} }
   child { node[inner] {}
    child { node[leaf] {} } } }
 child { node[inner] {}
  child { node[leaf] {} }
  child { node[inner] {}
   child { node[leaf] {} } } }
 child { node[inner] {}
  child { node[leaf] {} }
  child { node[inner] {}
   child { node[leaf] {} } } } }
 child { node[leaf] {} }
 child { node[leaf] {} };
\node[stage] at (0,.5) {4};
\etkz
\end{align*}
Though the number of heads can grow considerably in one step, it turns out
that the fight always terminates, and Hercules will win independent of his
strategy. Proving termination of the Battle is challenging since
Kirby and Paris proved in their landmark paper \cite{KP82} that
termination for an arbitrary (computable) strategy is independent of
Peano arithmetic.
In \cite{KP82} a termination argument based on ordinals is used.

Starting with \cite[p.~271]{DJ90}, several TRS encodings of the Battle of
Hercules and Hydra have been proposed and studied
\cite{B06,DM07,FZ96,M09,T98}.
Touzet~\cite{T98} was the first to give a
rigorous termination proof and in \cite{ZWM15} the automation of ordinal
interpretations is discussed.
In this article we present yet another
encoding. In contrast to earlier TRS encodings that model a specific
strategy, it uses AC matching to represent \emph{arbitrary} battles.
To prove its termination, we adapt existing termination
methods for AC rewriting.

The remainder of the article is organized as follows. After
recalling some basic definitions in \secref{preliminaries}, we present
our new encoding of the Battle in \secref{encoding}. We give a
rigorous proof that our encoding faithfully represents the Battle.
In \secref{semantic labeling} we present 
many-sorted semantic labeling for AC rewriting and apply it to our
encoding. This results in an infinite AC rewrite system, which
can be shown terminating by Rubio's AC-RPO~\cite{R02}.
As a matter of fact, we do not need the full power of AC-RPO. Inspired
by Steinbach's AC-KBO~\cite{S90}, in \secref{acmpo} we introduce AC-MPO,
a much weakened version of AC-RPO, and show that it is powerful enough
for our purpose. Some of the properties of AC-MPO are proved in the
appendix.

Related work is discussed in
\secref{related}. In particular, we comment on earlier encodings of the
Battle. We conclude in \secref{conclusion} with suggestions for future
research.

A preliminary version of this article appeared in the proceedings of the
8th International Conference on Formal Structures for Computation and
Deduction~\cite{HM23}. AC-MPO is a new result. New examples
provide further illustration of the simulation of the Battle of Hercules
and Hydra.

\section{Preliminaries}
\label{sec:preliminaries}

Let $\xS$ be a set of \emph{sorts}. An $\xS$-\emph{sorted signature} $\xF$
consists of function symbols $f$ having a sort declaration 
$S_1 \times \cdots \times S_n \to S$. Here $\seq{S}$ and $S$ are sorts in
$\xS$ and $n$ is the \emph{arity} of $f$. By $f^{(n)}$ we indicate
that $f$ has arity $n$. Let $\xV$ be a countably infinite set of
variables, where every variable has its own sort.
We assume the existence of infinitely many variables of each sort.
\emph{Terms} of sort $S$ are inductively defined as usual: Every
variable of sort $S$ is a term of sort $S$ and if $f$ has sort
declaration $S_1 \times \cdots \times S_n \to S$ and $t_i$ is a term of
sort $S_i$ for all $1 \leqslant i \leqslant n$ then $f(\seq{t})$ is a term
of sort $S$. 
\emph{Ground terms} are terms without variables.
The \emph{root symbol} $\m{root}(t)$ of a term $t$ is $t$ if it is a
variable, and $f$ if $t = f(\seq{t})$.
For every sort $S$ we introduce a fresh constant $\Box_S$, called the
\emph{hole}. A term over $\xF \uplus \{ \Box_S \mid S \in \xS \}$ is a
\emph{context} over $\xF$ if it contains exactly one hole. Given a
context $C$ and a term $t$, we write $C[t]$ for the term resulting from
replacing the hole in $C$ by $t$.
We write $s \subterm t$ if $t = C[s]$ for some context $C$.  We write
$s \prsubterm t$ if $s \subterm t$ and $s \neq t$.
A mapping $\sigma$ that associates each variable to a
term of the same sort is a \emph{substitution} if its domain
$\SET{x \in \xV \mid \sigma(x) \neq x}$ is finite. The application
$t\sigma$ of $\sigma$ to a term $t$ is defined as $\sigma(t)$ if $t$ is a
variable and $f(t_1\sigma,\dots,t_n\sigma)$ if $t = f(\seq{t})$.
A binary relation $\to$ on terms is \emph{closed under substitutions} if
$s\sigma \to t\sigma$ whenever $s \to t$, for all substitutions
$\sigma$. It is
\emph{closed under contexts} if $C[s] \to C[t]$ whenever 
$s \to t$, for all contexts $C$. 
It has the \emph{subterm property} if the inclusion
${\prsuperterm} \subseteq {\to}$ holds.
Moreover, the relation $\to$ is said to be a \emph{rewrite relation}
if it is closed under contexts and substitutions.
\emph{Rewrite orders} are rewrite relations that are strict orders,
and \emph{reduction orders} are rewrite orders that are well-founded.

A \emph{rewrite rule} $\ell \R r$ consists of two terms $\ell$ and
$r$ of the same sort such that all variables in $r$ occur in $\ell$.
A (many-sorted) \emph{term rewrite system} (TRS) is a set of rewrite
rules. We denote by $\Rb[\xR]$ the smallest rewrite relation that
contains the pairs of the TRS $\xR$.
A rule $\ell \R r$ is \emph{non-collapsing} if $r$ is
not a variable. A TRS is called \emph{non-collapsing} if all rules
are non-collapsing.
A TRS $\xR$ is \emph{terminating} if $\Rb[\xR]$ is well-founded.

Let $\xF_\AC$ be a subset of the binary function symbols in
$\xF$ that have sort declarations of the form $S \times S \to S$. We
denote by $\AC$ the set of equations
\begin{align*}
f(f(x,y),z) &\approx f(x,f(y,z)) & f(x,y) &\approx f(y,x)
\end{align*}
expressing the associativity and commutativity of each $f \in \xF_\AC$.
Since equations in $\AC$ are rewrite
rules, we can view $\AC$ as a TRS. Using this fact, we
define the relation $\ac$
as the reflexive, transitive, and symmetric closure of $\Rb[\AC]$.
Let $\xR$ be a TRS. The relation $\ac \cdot \Rb[\xR] \cdot \ac$
is called AC rewriting and abbreviated by $\to_{\xR/\AC}$. We
say that $\xR$ is \emph{AC terminating} if
$\to_{\xR/\AC}$ is well-founded.
A reduction order $>$ is \emph{AC-compatible} if the inclusion
${\ac \cdot > \cdot \ac} \subseteq {>}$ holds. AC termination of a
TRS $\xR$ can be shown by finding an AC-compatible reduction order such
that $\xR \subseteq {>}$ holds.

The above definitions specialize to the usual unsorted setting
when the set of sorts is a singleton set.

Finally, we recall two order extensions. Let $>$ be a strict order on a
set $A$. The \emph{lexicographic extension} $>^\lex$ of $>$ is defined on
tuples over $A$ as follows: $(\seq[m]{a}) >^\lex (\seq{b})$ if $n = m$ and
there exists an index $1 \leqslant k \leqslant n$ such that $a_k > b_k$
and $a_i = b_i$ for all $i < k$. The \emph{multiset extension} $>^\mul$ of
$>$ is defined on multisets over $A$ as follows: $M >^\mul N$ if there
exist multisets $X$ and $Y$ such that $N = (M - X) \uplus Y$, 
$\varnothing \neq X \subseteq M$, and every $b \in Y$ admits an element
$a \in X$ with $a > b$.

\section{Encoding}
\label{sec:encoding}

First we give a formal account of the Hydra Battle.

\begin{defi}
\label{def:hydra}
To represent Hydras, we use a signature containing a constant
symbol $\h$ representing a head, a binary symbol $\si$ for siblings, and
a unary function symbol $\I$ representing the internal nodes. We use infix
notation for $\si$ and declare it to be an \textup{AC} symbol.
We write $\xTH$ for the set of ground terms over $\SET{\h,\I,{\si}}$.
\emph{Encodings} of Hydras are terms $t$ in $\xTH$ with 
$\RT{t} \in \SET{\h,\I}$.
\end{defi}

To improve readability we omit parentheses in terms with nested
$\si$ symbols in examples.

\begin{exa}
\label{exa:hydras}
The Hydras in the above example fight are represented by the terms
\begin{alignat*}{3}
H_0 &= \I(\I(\h) \,\si\, \I(\I(\I(\h) \si \I(\h))) \,\si\, \h) \\
H_1 &= \I(\I(\h) \,\si\, \I(\I(\I(\h) \si \h \si \h)) \,\si\, \h) \\
H_2 &= \I(\I(\h) \,\si\, \I(\I(\I(\h) \si \h) \si \I(\I(\h) \si \h) \si
\I(\I(\h) \si \h)) \,\si\, \h) \\
H_3 &= \I(\h \,\si\, \h \,\si\, \h \,\si\, \I(\I(\I(\h) \si \h) \si
\I(\I(\h)
\si \h) \si \I(\I(\h) \si \h)) \,\si\, \h \,\si\, \h) \\
H_4 &= \I(\h \,\si\, \h \,\si\, \I(\I(\I(\h) \si \h) \si \I(\I(\h) \si \h)
\si \I(\I(\h) \si \h)) \,\si\, \h \,\si\, \h)
\end{alignat*}
and they are encodings of Hydras. The term $\h \si \h$ is included in
$\xTH$ but not regarded as an encoding of a Hydra.
\end{exa}

\begin{defi}
\label{def:Rn}
Let $n$ be a natural number. The \textup{TRS} $\xR_n$ operates on
encodings of Hydras and consists of the following four rules:
\begin{align*}
\I(\I(\h)) &\nR{1} \I(\h^{n+2}) &
\I(\I(\h) \si y) &\nR{3} \I(\h^{n+2} \si y) \\
\I(\I(\h \si x)) &\nR{2} \I(\I(x)^{n+2}) &
\I(\I(\h \si x) \si y) &\nR{4} \I(\I(x)^{n+2} \si y)
\end{align*}
Here $t^k$ for $k \geqslant 1$ is defined inductively as follows:
\begin{align*}
t^k &= \begin{cases}
t &\text{if $k = 1$} \\
t^{k-1} \si t &\text{if $k > 1$}
\end{cases}
\end{align*}
The transition relation $\Rightarrow_n$ on encodings of Hydras is
defined as follows: $H \Rightarrow_n H'$ if
\begin{enumerate}
\Item
$H = \I(\h)$ and $H' = \h$, or
\Item
$H \ac \I(\h \si t)$ and $H' = \I(t)$ for some term $t$, or
\Item
$H \RRAC H'$.
\end{enumerate}
\end{defi}

That $H$ and $H'$ are encodings of successive Hydras at stages $n$ and
$n+1$ in a battle is expressed as $H \Rightarrow_n H'$. So fights with
Hydras are represented by finite or infinite sequences of the form
$H_0 \Rightarrow_0 H_1 \Rightarrow_1 H_2 \Rightarrow_2 \cdots$.

\begin{exa}[continued from \exaref{hydras}]
We have the following sequence:
\[
H_0 \,\Rightarrow_0\, H_1 \,\Rightarrow_1\, \cdots \,\Rightarrow_3\, H_4
\]
For instance, the first step is verified as follows: Let $\ell \to r$
be the third rule in $\xR_0$, and let
$C = \I(\I(\h) \,\si\, \I(\Box) \,\si\, \h)$ and
$\sigma = \SET{y \mapsto \I(\h)}$. Since $H_0 = C[\ell\sigma]$ and
$H_1 =_\AC C[r\sigma]$ hold, we have $H_0 \RRAC[0] H_1$ and so
$H_0 \Rightarrow_0 H_1$ is obtained.
\end{exa}

Now we present our TRS encoding of the Hydra Battle. We represent
natural numbers $n$ by $\m{s}^n(\m{0})$, which is abbreviated to $\m{n}$.

\begin{defi}
Let $\xF$ be the signature consisting of the constant $\m{0}$, the
unary symbol $\m{s}$, the five binary symbols $\m{A}\,\text{--}\,\m{E}$ as
well as the three symbols $\h$, $\I$, and $\si$ in \defref{hydra}.
The \textup{TRS} $\xH$ over $\xF$ consists of the following 14
rewrite rules:
\begin{align*}
\m{A}(n,\I(\h)) &\nR{1} \m{A}(\m{s}(n),\h) &
\m{D}(n,\I(\I(x))) &\nR{8} \I(\m{D}(n,\I(x))) \\
\m{A}(n,\I(\h \si x)) &\nR{2} \m{A}(\m{s}(n),\I(x)) &
\m{D}(n,\I(\I(x) \si y)) &\nR{9} \I(\m{D}(n,\I(x)) \si y) \\
\m{A}(n,\I(x)) &\nR{3} \m{B}(n,\m{D}(\m{s}(n),\I(x))) &
\m{D}(n,\I(\I(\h \si x) \si y)) &\nR{10} \I(\m{C}(n,\I(x)) \si y) \\
\m{C}(\m{0},x) &\nR{4} \m{E}(x) &
\m{D}(n,\I(\I(\h \si x))) &\nR{11} \I(\m{C}(n,\I(x))) \\
\m{C}(\m{s}(n),x) &\nR{5} x \si \m{C}(n,x) &
\m{D}(n,\I(\I(\h) \si y)) &\nR{12} \I(\m{C}(n,\h) \si y) \\
\I(\m{E}(x) \si y) &\nR{6} \m{E}(\I(x \si y)) &
\m{D}(n,\I(\I(\h))) &\nR{13} \I(\m{C}(n,\h)) \\
\I(\m{E}(x)) &\nR{7} \m{E}(\I(x)) &
\m{B}(n,\m{E}(x)) &\nR{14} \m{A}(\m{s}(n),x)
\end{align*}
\end{defi}

The Battle is started with the term $\m{A}(\m{0},H)$ where $H$
is the encoding of the initial Hydra. Rule 1 takes care of the dying Hydra
\btkz[ld=7,level 2/.style={sd},rotate=90]
\node[inner] at (0,0) {}
 child { node[leaf] {} };
\etkz\,. 
An application of the rule ends the battle with a term of
the form $\m{A}(n,\h)$. 
The $\h$ denotes here a dead Hydra; a Hydra with
only one head is represented by the term $\I(\h)$.
Rule 2 cuts a head without grandparent node, and so no
copying takes place. Due to the power of AC matching, the removed
head need not be the leftmost one. With rule 3, the search for locating a
head with grandparent node starts. The search is performed with the
auxiliary symbol $\m{D}$ and involves rules 8--13. When the head
to be cut is located (in rules 10--13), copying begins with the
auxiliary symbol $\m{C}$ and rules 4 and 5. The end of the copying phase
is signaled with $\m{E}$, which travels upwards with rules 6 and 7.
Finally, rule 14 creates the next stage of the Battle. Note that we make
extensive use of AC matching to simplify the search process.

\begin{thm}
\label{thm:simulation}
Let $n$ be a natural number. If $H \Rightarrow_n H'$ then
$\m{A}(\m{n},H) \RHAC[+] \m{A}(\msn,H')$.
\end{thm}

Before presenting the proof, we illustrate how 
AC rewriting of $\xH$ simulates fights with Hydras.

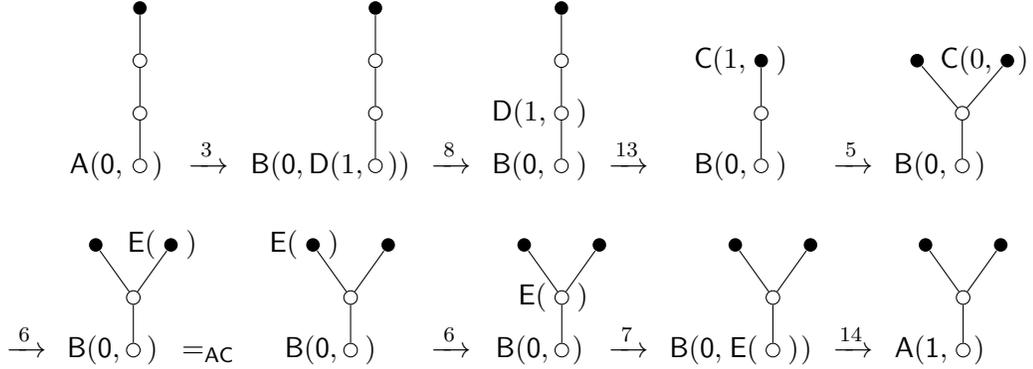
\begin{figure}
\[
\begin{array}{c@{\;\;}c@{\;\;}c@{\;\;}c@{\;\;}c@{\;\;}c@{\;\;}c@{\;\;}%
c@{\;\;}c@{\;\;}c}
&
\m{A}(\m{0},\,
\begin{tikzpicture}[H,baseline=(1.south)]
\node[inner] (1) at (0,0) {}
 child { node[inner] {} 
  child { node[inner] {} 
   child { node[leaf] {} } } };
\end{tikzpicture}\,)
&\nR{3}&
\m{B}(\m{0},\m{D}(1,
\begin{tikzpicture}[H,baseline=(1.south)]
\node[inner] (1) at (0,0) {}
 child { node[inner] (2) {} 
  child { node[inner] (3) {} 
   child { node[leaf] (4) {} } } };
\end{tikzpicture}\,))
&\nR{8}&
\m{B}(\m{0},\hh[-9]
\begin{tikzpicture}[H,baseline=(1.south)]
\node[inner] (1) at (0,0) {}
 child { node[inner] (2) {} 
  child { node[inner] (3) {} 
   child { node[leaf] (4) {} } } };
\node[left]  at (2) {$\m{D}(1,$\,};
\node[right] at (2) {\,$)$};
\end{tikzpicture}\hh[-2.8])
&\nR{13}&
\m{B}(\m{0},\hh[-9]
\begin{tikzpicture}[H,baseline=(1.south)]
\node[inner] (1) at (0,0) {}
 child { node[inner] (2) {} 
  child { node[leaf] (3) {} } } ;
\node[left] at (3) {$\m{C}(1,\,$};
\node[right] at (3) {\,$)$};
\end{tikzpicture}\hh[-2.8])
&\nR{5}&
\m{B}(\m{0},\hh[-5.5]
\begin{tikzpicture}[H,baseline=(1.south),
level 2/.style={sibling distance=12mm}]
\node[inner] (1) at (0,0) {}
 child { node[inner] (2) {} 
  child { node[leaf] (5) {} }
  child { node[leaf] (3) {} } } ;
\node[left] at (5) {$\m{C}(0,\,$};
\node[right] at (5) {$)$};
\end{tikzpicture}\hh[-8.5])
\\[3ex]
\nR{6}&
\m{B}(\m{0},\hh[-4.8]
\begin{tikzpicture}[H,baseline=(1.south),
level 2/.style={sibling distance=10mm}]
\node[inner] (1) at (0,0) {}
 child { node[inner] (2) {} 
  child { node[leaf] (5) {} }
  child { node[leaf] (3) {} }
 }
;
\node[left] at (5) {$\m{E}(\,$};
\node[right] at (5) {$\,)$};
\end{tikzpicture}\hh[-8])\;
&\ac&
\m{B}(\m{0},\hh[-11]
\begin{tikzpicture}[H,baseline=(1.south),
level 2/.style={sibling distance=10mm}]
\node[inner] (1) at (0,0) {}
 child { node[inner] (2) {} 
  child { node[leaf] (5) {} }
  child { node[leaf] (3) {} }
 }
;
\node[left] at (3) {$\m{E}(\,$};
\node[right] at (3) {$\,)$};
\end{tikzpicture}\hh[-4])
&\nR{6}&
\m{B}(\m{0},\hh[-6]
\begin{tikzpicture}[H,baseline=(1.south),
level 2/.style={sibling distance=10mm}]
\node[inner] (1) at (0,0) {}
 child { node[inner] (2) {} 
  child { node[leaf] (5) {} }
  child { node[leaf] (3) {} }
 }
;
\node[left] at (2) {$\m{E}(\,$};
\node[right] at (2) {$\,)$};
\end{tikzpicture}\hh[-4.5])
&\nR{7}&
\m{B}(\m{0},\hh[-1]
\begin{tikzpicture}[H,baseline=(1.south),
level 2/.style={sibling distance=10mm}]
\node[inner] (1) at (0,0) {}
 child { node[inner] (2) {} 
  child { node[leaf] (5) {} }
  child { node[leaf] (3) {} }
 }
;
\node[left] at (1) {$\m{E}(\,$};
\node[right] at (1) {$\,)$};
\end{tikzpicture}\hh[-2.2])
&\nR{14}&
\m{A}(\m{1},\hh[-4.5]
\begin{tikzpicture}[H,baseline=(1.south),
level 2/.style={sibling distance=10mm}]
\node[inner] (1) at (0,0) {}
 child { node[inner] (2) {} 
  child { node[leaf] (5) {} }
  child { node[leaf] (3) {} }
 }
;
\end{tikzpicture}\hh[-4.5])
\end{array}
\]
\caption{Rewriting from $\m{A}(\m{0},\I(\I(\I(\h))))$ to
$\m{A}(\m{1},\I(\I(\h \si \h)))$.}
\label{fig:illustration}
\end{figure}
\begin{exa}
\label{exa:fight}
Consider a fight with the Hydra of shape $\I(\I(\I(\h)))$. The fight
starts with the transition from $\I(\I(\I(\h)))$ to $\I(\I(\h \si \h))$.
This is simulated by the rewrite sequence
\begin{alignat*}{4}
&
&&\m{A}(\m{0},\I(\I(\I(\h))))
&&\nR{3} \m{B}(\m{0},\m{D}(\m{s}(\m{0}),\I(\I(\I(\h)))))
&&\nR{8} \m{B}(\m{0},\I(\m{D}(\m{s}(\m{0}),\I(\I(\h))))) \\
&\nR{13}
&&\m{B}(\m{0},\I(\I(\m{C}(\m{s}(\m{0}),\h))))
&&\nR{5} \m{B}(\m{0},\I(\I(\h \si \m{C}(\m{0},\h)))) 
&&\nR{4} \m{B}(\m{0},\I(\I(\h \si \m{E}(\h)))) \\
&\ac{}
&&\m{B}(\m{0},\I(\I(\m{E}(\h) \si \h)))
&&\nR{6} \m{B}(\m{0},\I(\m{E}(\I(\h \si \h))))
&&\nR{7} \m{B}(\m{0},\m{E}(\I(\I(\h \si \h)))) \\
&\nR{14}
&&\m{A}(\m{s}(\m{0}),\I(\I(\h \si \h))) 
\end{alignat*}
which is visualized in \figref{illustration}.
Rules 9--12 are variations of 8 and 13, which are used for handling
nodes that have siblings. To illustrate these, consider the first
step $H_0 \Rightarrow_0 H_1$
in the example fight in the introduction. The step is simulated by the
following rewrite sequence:
{\allowdisplaybreaks\begin{align*}
\m{A}(\m{0},H_0)
&\nR{3} \m{B}(\m{0},\m{D}(\m{s}(\m{0}),H_0)) \\
\ac \cdot\!
&\nR{9} \m{B}(\m{0},\I(\m{D}(\m{s}(\m{0}),\I(\I(\I(\h) \si \I(\h)))) \si
\I(\h) \si \h)) \\
&\nR{8} \m{B}(\m{0},\I(\I(\m{D}(\m{s}(\m{0}),\I(\I(\h) \si \I(\h)))) \si
\I(\h) \si \h)) \\
&\nR{12} \m{B}(\m{0},\I(\I(\I(\m{C}(\m{s}(\m{0}),\h) \si \I(\h))) \si
\I(\h) \si \h)) \\
&\nR{5} \m{B}(\m{0},\I(\I(\I(\h \si \m{C}(\m{0},\h) \si \I(\h))) \si
\I(\h) \si \h)) \\
&\nR{4} \m{B}(\m{0},\I(\I(\I(\h \si \m{E}(\h) \si \I(\h))) \si \I(\h)
\si \h)) \\
\ac \cdot\!
&\nR{6} \m{B}(\m{0},\I(\I(\m{E}(\I(\h \si \h \si \I(\h)))) \si \I(\h)
\si \h)) \\
&\nR{7} \m{B}(\m{0},\I(\m{E}(\I(\I(\h \si \h \si \I(\h)))) \si \I(\h)
\si \h)) \\
&\nR{6} \m{B}(\m{0},\m{E}(\I(\I(\I(\h \si \h \si \I(\h)))
\si \I(\h) \si \h))) \\
&\nR{14} \m{A}(\m{s}(\m{0}),\I(\I(\I(\h \si \h \si \I(\h)))
\si \I(\h) \si \h)) \,\ac\, \m{A}(\m{s}(\m{0}),H_1)
\end{align*}}
\end{exa}

It is important to note that the TRS $\xH$ defined above is
\emph{unsorted} and we establish in this article the result that it is AC
terminating on all terms. When simulating a battle, like in the statement
of the \thmref{simulation}, we deal with well-behaved terms adhering to
the sort discipline introduced shortly. The restriction to sorted terms is
crucial for our termination proof, but entails no loss of generality. This
is due to the following result, which is a special case of
\cite[Corollary~3.9]{MO00}.

\begin{thm}
\label{thm:persistency}
A non-collapsing \textup{TRS} over a many-sorted signature is
\textup{AC} terminating if and only if the corresponding \textup{TRS} over
the unsorted version of the signature is \textup{AC} terminating.
\qed
\end{thm}

The idea of using sorts to simplify termination proof goes back to
Zantema~\cite{Z94}.
The TRS $\xH$ can be seen as a TRS over the many-sorted signature $\xF'$: 
\begin{align*}
\h &: \m{O} &
\I, \m{E} &: \m{O} \to \m{O} &
{\si} &: \m{O} \times \m{O} \to \m{O} &
\m{A}, \m{B} &: \m{N} \times \m{O} \to \m{S} \\
\m{0} &: \m{N} &
\m{s} &: \m{N} \to \m{N} &&&
\m{C}, \m{D} &: \m{N} \times \m{O} \to \m{O}
\end{align*}
where $\m{N}$, $\m{O}$ and $\m{S}$ are sort symbols. Since $\xH$ is
non-collapsing, \thmref{persistency} guarantees that AC termination of
$\xH$ follows from AC termination of well-sorted terms over $\xF'$.

In the remainder of this section we present a proof of
\thmref{simulation} and its converse.

\begin{lem}
\label{lem:copy}
If $n > 0$ then $\m{C}(\m{n},t) \RHAC t^n \si \m{E}(t)$ for all
terms $t$.
\end{lem}

\begin{proof}
We use induction on $n$. If $n = 1$ then
\[
\m{C}(\m{n},t) \nR{5} t \si \m{C}(\m{0},t) \nR{6} t \si \m{E}(t) =
t^n \si \m{E}(t)
\]
Suppose the result holds for $n \geqslant 1$ and consider $n + 1$. The
induction hypothesis yields $\m{C}(\m{n},t) \RHAC t^n \si \m{E}(t)$. Hence
\[
\m{C}(\msn,t) \nR{5} t \si \m{C}(\m{n},t) \RHAC t \si
(t^n \si \m{E}(t)) \ac t^{n+1} \si \m{E}(t)
\qedhere
\]
\end{proof}

\begin{lem}
\label{lem:cut}
Let $n$ be a natural number. If $H \Rightarrow_n H'$
then $\m{D}(\msn,H) \RHAC \m{E}(H')$.
\end{lem}

\begin{proof}
We use structural induction on $H$ and consider the following two cases. 
\begin{itemize}
\item
First suppose $H \RRAC H'$ is a root step. If the first rule of $\xR_n$ is
used then $H = \I(\I(\h))$ and $H' \ac \I(h^{n+2})$. We have
$\m{D}(\msn,H) \nR{13} \I(\m{C}(\msn,\h))$. Using
\lemref{copy} we obtain
\begin{align*}
\I(\m{C}(\msn,\h)) &\RHAC \I(\h^{n+1} \si \m{E}(\h))
\ac \cdot \nR{6} \m{E}(\I(\h \si \h^{n+1})) \ac \m{E}(H')
\end{align*}
If the second rule of $\xR_n$ is used then $H \ac \I(\I(\h \si t))$ and
$H' \ac \I(\I(t)^{n+2})$ for some term $t$. We have
$\m{D}(\msn,H) \ac \cdot \nR{11} \I(\m{C}(\msn,\I(t)))$. Using
\lemref{copy} we obtain
\[
\I(\m{C}(\msn,\I(t))) \RHAC \I(\I(t)^{n+1} \si \m{E}(\I(t)))
\ac \cdot \nR{6} \m{E}(\I(\I(t) \si \I(t)^{n+1})) \ac \m{E}(H')
\]
If the third rule of $\xR_n$ is used then
$H \ac \I(\I(\h) \si t)$ and $H' \ac \I(\h^{n+2} \si t)$ for
some term $t$. We have
$\m{D}(\msn,H) \ac \cdot \nR{12}
\I(\m{C}(\msn,\h) \si t)$. The
remaining
argument is the same as in the preceding cases. If the fourth
rule of $\xR_n$ is used then $H \ac \I(\I(\h \si s) \si t)$ and
$H' \ac \I(\I(s)^{n+2} \si t)$ for some terms $s$ and $t$. Using
\lemref{copy} we obtain
\begin{align*}
\m{D}(\msn,H) &\ac \cdot \nR{10} \I(\m{C}(\msn,\I(s)) \si t)
\RHAC \I((\I(s)^{n+1} \si \m{E}(\I(s))) \si t) \\
&\ac \cdot \nR{6} \m{E}(\I(\I(s) \si (\I(s)^{n+1} \si t))) \ac
\m{E}(H')
\end{align*}
\item
Otherwise, $H \ac \I(H_1 \si H_2 \si \cdots \si H_m)$ and
$H' \ac \I(H_1' \si H_2 \si \cdots \si H_m)$ for some $m \geqslant 1$
and Hydras $\seq[m]{H}, H_1'$ with $H_1 \RRAC H_1'$. We obtain
$\m{D}(\msn,H_1) \RHAC \m{E}(H_1')$ from the induction hypothesis.
Note that $\m{root}(H_1) = \I$. If $m = 1$ then
\begin{align*}
\m{D}(\msn,H) &\ac \m{D}(\msn,\I(H_1))
\nR{8} \I(\m{D}(\msn,H_1)) \RHAC \I(\m{E}(H_1'))
\nR{7} \m{E}(\I(H_1')) \\
&\ac \m{E}(H')
\end{align*}
and if $m > 1$ we reach the same conclusion using rules 9 and 6 instead
of 8 and 7.
\qedhere
\end{itemize}
\end{proof}

\begin{proof}[Proof of \thmref{simulation}]
Our task is to show
\[
\m{A}(\m{n},H) \RHAC \m{A}(\msn,H')
\]
If $H \Rightarrow_n H'$ is derived from condition (1) or (2) in
\defref{Rn}, the claim is immediate by rules 1 and 2 of $\xH$. Otherwise,
$H \RRAC H'$. This implies $\m{root}(H) = \I$. Using rules $3$ and
$14$ together with \lemref{cut} yields
\[
\m{A}(\m{n},H) \nR{3} \m{B}(\m{n},\m{D}(\msn,H))
\RHAC \m{B}(\m{n},\m{E}(H')) \nR{14} \m{A}(\msn,H')
\qedhere
\]
\end{proof}

In the remaining part of this section we prove the converse of
\thmref{simulation}.

\begin{thm}
\label{thm:converse}
Let $H, H' \in \xTH$ be encodings of Hydras and let $n$ be a natural
number. If $\m{A}(\m{n},H) \RHAC \m{A}(\msn,H')$ then
$H \Rightarrow_n H'$.
\end{thm}

In order to show the claim we need a few auxiliary lemmata.
Let $\xCH$ be the set of ground contexts over $\SET{\h,\I,{\si}}$.

\begin{defi}
We define $U$ as the set consisting of all terms of the forms
$\m{A}(\m{n},t)$, 
$\m{B}(\m{n},C[\m{C}(\m{m},t)])$,
$\m{B}(\m{n},C[\m{D}(\msn,t)])$, and
$\m{B}(\m{n},C[\m{E}(t)])$,
where $n, m \in \NN$, $t \in \xTH$, and $C \in \xCH$.
\end{defi}

The set $U$ contains all terms reachable from $\m{A}(\m{n},H)$.

\begin{lem}
\label{lem:U}
If $t \in U$ and $t \Rab[\xH\hh\cup\hh\AC]{*} u$ then $u \in U$.
\end{lem}
\begin{proof}
The claim is easily shown by induction on the length of 
$t \Rab[\xH\hh\cup\hh\AC]{*} u$.
\end{proof}

In order to analyze the rewrite sequence
$\m{A}(\m{n},H) \RHAC \m{A}(\msn,H')$ we define three subsets
of $\xH$: $\xH_1 = \SET{1,2}$, $\xH_2 = \SET{3\,\text{--}\,9,14}$, and
$\xH_3 = \SET{10\,\text{--}\,13}$. The second rewrite sequence in
\exaref{fight} can then be described as follows:
\begin{align*}
\m{A}(\m{0},H_0)
&\RHnAC{2} \m{B}(\m{0},\I(\I(\m{D}(\m{s}(\m{0}),\I(\I(\h) \si \I(\h))))
\si \I(\h) \si \h)) \\
&\RHnAC[]{3} \m{B}(\m{0},\I(\I(\I(\m{C}(\m{s}(\m{0}),\h) \si \I(\h))) \si
\I(\h) \si \h)) \\
&\RHnAC{2} \m{A}(\m{1},H_1)
\end{align*}

\begin{defi}
We define $V$ as the extension of $U$ with $\xTH$ and all terms of the
forms $C[\m{C}(\m{n},t)]$ $C[\m{D}(\m{n},t)]$, and $C[\m{E}(t)]$
where $n \in \NN$, $t \in \xTH$, and $C \in \xCH$.
The mapping $\pi\colon V \to \xTH$ is defined as follows:
\[
\pi(t) = \begin{cases}
\h &\text{if $t = \h$} \\
\I(\pi(u)) &\text{if $t = \I(u)$} \\
\pi(u) \si \pi(v) &\text{if $t = u \si v$} \\
u &\text{if $t = \m{A}(\m{n},u)$ or $t = \m{D}(\m{n},u)$ or
$t = \m{E}(u)$} \\
\pi(u) &\text{if $t = \m{B}(\m{n},u)$} \\
u^{n+1} &\text{if $t = \m{C}(\m{n},u)$}
\end{cases}
\]
\end{defi}

Taking the role of $\m{C}$ into account, the mapping $\pi$ computes the
Hydra in a given term. Applying $\pi$ to the terms in the above rewrite
sequence of $\xH_2/\AC$ and $\xH_3/\AC$, we obtain
\begin{align*}
H_0 = \pi(\m{A}(\m{0},H_0)) \ac {}&
\pi(\m{B}(\m{0},\I(\I(\m{D}(\m{s}(\m{0}),\I(\I(\h) \si \I(\h)))) \si
\I(\h) \si \h))) \\
\RRAC[0] {}&
\pi(\m{B}(\m{0},\I(\I(\I(\m{C}(\m{s}(\m{0}),\h) \si \I(\h))) \si \I(\h)
\si \h))) \\
\ac {}& \pi(\m{A}(\m{1},H_1)) = H_1
\end{align*}
This verifies that $H_1$ is a successor of $H_0$. 

\begin{lem}
\label{lem:pi}
The following properties hold.
\begin{enumerate}
\item
$\pi(t) = t$ for all terms $t \in \xTH$,
\item
$\pi(C[t]) = C[\pi(t)]$ for all terms $t \in V$ and contexts
$C \in \xCH$,
\item
$\pi(C[t]) \ac \pi(D[u])$ for all terms $t, u \in \xTH$ and
contexts $C, D \in \xCH$ with $t \ac u$ and $C \ac D$.
\end{enumerate}
\end{lem}

\begin{proof}
The first statement is proved by induction on $t \in \xTH$. If $t = \h$
then $\pi(t) = \h = t$. If $t = \I(u)$ with $u \in \xTH$ then
$\pi(t) = \I(\pi(u)) = \I(u) = t$. If $t = u \si v$ with $u, v \in \xTH$
then $\pi(t) = \pi(u) \si \pi(v) = u \si v = t$.
For the second statement we use induction on the context $C \in \xCH$.
If $C = \Box$ then $\pi(C[t]) = \pi(t) = C[\pi(t)]$. If $C = \I(D)$ then
$\pi(C[t]) = \I(\pi(D[t])) = \I(D[\pi(t)]) = C[\pi(t)]$. If $C = D \si u$
then $D \in \xCH$ and $u \in \xTH$ and thus
$\pi(C[t]) = \pi(D[t]) \si \pi(u) = D[\pi(t)] \si u = C[\pi(t)]$. If
$C = u \si D$ then $D \in \xCH$ and $u \in \xTH$ and thus
$\pi(C[t]) = \pi(u) \si \pi(D[t]) = u \si D[\pi(t)] = C[\pi(t)]$.
The third statement follows from statements (1) and (2):
$\pi(C[t]) = C[\pi(t)] = C[t] \ac D[t] \ac D[u] = D[\pi(u)] = \pi(D[u])$.
\end{proof}

The following lemma relates AC rewriting of $\xH$ to rewriting of Hydras
according to \defref{Rn}.

\begin{lem}
\label{lem:projection}
The following statements hold for all terms $s, t \in U$.
\begin{enumerate}
\item
If $s \ac t$ then $\pi(s) \ac \pi(t)$.
\item
If $s \Rb[\xH_2] t$ then $\pi(s) \ac \pi(t)$.
\item
If $s \Rb[\xH_3] t$ then $\pi(s) \RRAC \pi(t)$
with $s = \m{B}(\m{n},s')$ for some $n \geqslant 0$.
\end{enumerate}
\end{lem}
\begin{proof}
Let $s, t \in U$.
\begin{enumerate}
\item
If $s = \m{A}(\m{n},u)$ with
$u \in \xTH$ then $t = \m{A}(\m{n},v)$ for some term $v \in \xTH$ with
$u \ac v$. Since $\pi(s) = u$ and $\pi(t) = v$, $\pi(s) \ac \pi(t)$
follows. If $s = \m{B}(\m{n},C[\m{C}(\m{m},u)])$ with $n, m \in \NN$,
$C \in \xCH$ and $u \in \xTH$ then $t = \m{B}(\m{n},D[\m{C}(\m{m},v)])$
with $C \ac D$ and $u \ac v$. Using \lemref{pi}(1,2) we obtain
$\pi(s) = \pi(C[\m{C}(\m{m},u)]) = C[\pi(\m{C}(\m{m},u))] = C[u^{m+1}]$
and $\pi(t) = D[v^{m+1}]$. From $u \ac v$ we infer $u^{m+1} \ac v^{m+1}$
and thus $\pi(s) \ac \pi(t)$ by \lemref{pi}(3). The cases
$s = \m{B}(\m{n},C[\m{D}(\msn,u)])$
and $s = \m{B}(\m{n},C[\m{E}(u)])$ are treated in the same way.
\item
For the second statement we make a case analysis based on the employed
rule in $\xH_2$.
\begin{itemize}
\item
If $s \nR{3} t$ then $s = \m{A}(\m{n},\I(u))$ and
$t = \m{B}(\m{n},\m{D}(\msn,\I(u)))$ for some $n \geqslant 0$ and
$u \in \xTH$. We have
$\pi(s) = \I(u) = \pi(\m{D}(\msn,\I(u))) = \pi(t)$ by the definition
of $t$.
\item
If $s \nR{4} t$ then $s = \m{B}(\m{n},C[\m{C}(\m{0},u)])$ and
$t = \m{B}(\m{n},C[\m{E}(u)])$ for some $n \geqslant 0$, $C \in \xCH$ and
$u \in \xTH$. We have
$\pi(s) = \pi(C[\m{C}(\m{0},u)]) = C[u^1] = C[u] = \pi(C[u]) = \pi(t)$.
\item
If $s \nR{5} t$ then $s = \m{B}(\m{n},C[\m{C}(\msm,u)])$ and
$t = \m{B}(\m{n},C[u \si \m{C}(\m{m},u)])$ for some $n \geqslant 0$,
$m \geqslant 0$,
$C \in \xCH$ and $u \in \xTH$. We have $\pi(s) = C[u^{m+2}] \ac
C[u \si u^{m+1}] = C[\pi(u \si u^{m+1})] = 
\pi(C[u \si \m{C}(\m{m},u)]) = \pi(t)$.
\item
If $s \nR{6} t$ then $s = \m{B}(\m{n},C[\I(\m{E}(u) \si v)])$ and
$t = \m{B}(\m{n},C[\m{E}(\I(u \si v))])$ for some $n \geqslant 0$,
$C \in \xCH$ and $u, v \in \xTH$. We have
$\pi(s) = \pi(C[\I(\m{E}(u) \si v)]) = C[\I(u \si v)] =
\pi(C[\m{E}(\I(u \si v))]) = \pi(t)$.
\item
If $s \nR{7} t$ then $s = \m{B}(\m{n},C[\I(\m{E}(u))])$ and
$t = \m{B}(\m{n},C[\m{E}(\I(u))])$ for some $n \geqslant 0$, $C \in \xCH$
and $u \in \xTH$. We have $\pi(s) = \pi(C[\I(\m{E}(u))]) =
C[\I(u)] = \pi(C[\m{E}(\I(u))]) = \pi(t)$.
\item
If $s \nR{8} t$ then $s = \m{B}(\m{n},C[\m{D}(\msn,\I(\I(u)))])$ and
$t = \m{B}(\m{n},C[\I(\m{D}(\msn,\I(u)))])$ for some
$n \geqslant 0$, $C \in \xCH$ and $u \in \xTH$. We have
$\pi(s) = C[\I(\I(u))] = \pi(t)$.
\item
If $s \nR{9} t$ then $s = \m{B}(\m{n},C[\m{D}(\msn,\I(\I(u) \si v))])$
and $t = \m{B}(\m{n},C[\I(\m{D}(\msn,\I(u)) \si v)])$ for some
$n \geqslant 0$, $C \in \xCH$ and $u, v \in \xTH$. In this case we
obtain $\pi(s) = C[\I(\I(u) \si v)] = \pi(t)$.
\item
If $s \nR{14} t$ then $s = \m{B}(\m{n},\m{E}(u))$ and
$t = \m{A}(\msn,u)$ for some $n \geqslant 0$ and $u \in \xTH$. In this
case we have $\pi(s) = \pi(\m{E}(u)) = u = \pi(t)$.
\end{itemize}
\item
Again we make a case analysis on the applied rewrite rule.
\begin{itemize}
\item
If $s \nR{10} t$ then
$s = \m{B}(\m{n},C[\m{D}(\msn,\I(\I(\h \si u) \si v))])$ and
$t = \m{B}(\m{n},C[\I(\m{C}(\msn,\I(u)) \si v)])$ for some
$n \geqslant 0$, $C \in \xCH$ and $u, v \in \xTH$. We obtain
$\pi(s) = C[\I(\I(\h \si u) \si v)]$ and
$\pi(t) = C[\I(\I(u)^{n+2} \si v)]$. Hence $\pi(s) \Rb[\xR_n] \pi(t)$ by
applying rule 4 of $\xR_n$.
\item
If $s \nR{11} t$ then
$s = \m{B}(\m{n},C[\m{D}(\msn,\I(\I(\h \si u)))])$ and
$t = \m{B}(\m{n},C[\I(\m{C}(\msn,\I(u)))])$ for some $n \geqslant 0$,
$C \in \xCH$ and $u, v \in \xTH$. We obtain $\pi(s) = C[\I(\I(\h \si u))]$
and $\pi(t) = C[\I(\I(u)^{n+2})]$. Hence $\pi(s) \Rb[\xR_n] \pi(t)$ by
applying rule 2 of $\xR_n$.
\item
If $s \nR{12} t$ then
$s = \m{B}(\m{n},C[\m{D}(\msn,\I(\I(\h) \si v))])$ and
$t = \m{B}(\m{n},C[\I(\m{C}(\msn,\h) \si v)])$ for some
$n \geqslant 0$, $C \in \xCH$ and $v \in \xTH$. We obtain
$\pi(s) = C[\I(\I(\h) \si v)]$ and
$\pi(t) = C[\I(\h^{n+2} \si v)]$. Hence $\pi(s) \Rb[\xR_n] \pi(t)$ by
applying rule 3 of $\xR_n$.
\item
If $s \nR{13} t$ then $s = \m{B}(\m{n},C[\m{D}(\msn,\I(\I(\h)))])$
and $t = \m{B}(\m{n},C[\I(\m{C}(\msn,\h))])$ for some $n \geqslant 0$
and $C \in \xCH$. We obtain $\pi(s) = C[\I(\I(\h))]$ and
$\pi(t) = C[\I(\h^{n+2})]$. Hence $\pi(s) \Rb[\xR_n] \pi(t)$ by applying
rule 1 of $\xR_n$.
\qedhere
\end{itemize}
\end{enumerate}
\end{proof}

So we are ready to prove the main claim.

\begin{proof}[Proof of \thmref{converse}]
Suppose $s = \m{A}(\m{n},H) \RHAC[+] \m{A}(\msn,H') = t$.
Inspection of $\xH$ reveals that one of the following two cases holds:
\begin{enumerate}[(a)]
\Item
\label{a}
$s \RHnAC[]{1} t$, or
\Item
\label{b}
$s \RHnAC{2} \cdot \RHnAC[]{3} \cdot \RHnAC{2} t$.
\end{enumerate}
\smallskip
We first consider \ref{a}.
If $s \RHnAC[]{1} t$ is a root step using rule 1 then
$H = \I(\h)$ and $H' = \h$.
If $s \RHnAC[]{1} t$ is a root step using rule 2 then
$H \ac \I(\h \si u)$ and $H' \ac \I(u)$ for some term $u$.
Next we consider \ref{b}. We have
$s \RHnAC{2} s' \RHnAC[]{3} t' \RHnAC{2} t$ for some $s'$ and $t'$. From
\lemref{U} we obtain $s, s', t', t \in U$. Hence
\[
H = \pi(s) \ac \pi(s') \RRAC \pi(t') \ac \pi(t) = H'
\]
is obtained by \lemref{projection} and thus 
$H \RRAC H'$. 
Hence, we conclude $H \Rightarrow_n H'$.
\end{proof}

\section{Many-Sorted Semantic Labeling modulo AC}
\label{sec:semantic labeling}

Kirby and Paris~\cite{KP82} proved the termination of the Hydra Battle
by associating ordinal numbers to Hydras (see~\cite{KM76,DM07} for
notions and notations for ordinal numbers). Consider, for example, the
following fight with the Hydra in \exaref{fight}:
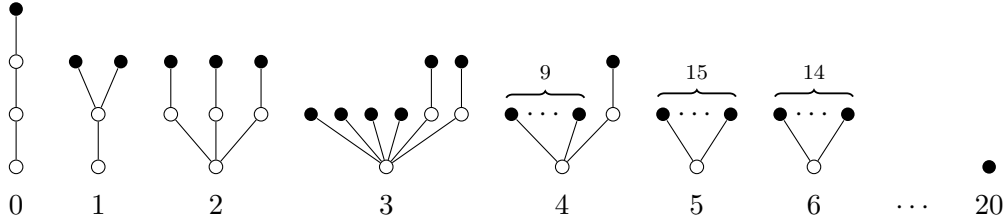
\begin{align*}
&
\btkz[H,baseline=(0.base)]
\node[inner] (1) at (0,0) {}
 child { node[inner] {} 
  child { node[inner] {} 
   child { node[leaf] {} } } };
\node[stage] (0) at (0,.5) {0};
\etkz
&&
\btkz[H,baseline=(0.base)]
\node[inner] (1) at (0,0) {}
 child { node[inner] {} 
   child { node[leaf] {} } 
   child { node[leaf] {} } };
\node[stage] (0) at (0,.5) {1};
\etkz
&&
\btkz[H,baseline=(0.base)]
\node[inner] (1) at (0,0) {}
 child { node[inner] {} child { node[leaf] {} } }
 child { node[inner] {} child { node[leaf] {} } }
 child { node[inner] {} child { node[leaf] {} } };
\node[stage] (0) at (0,.5) {2};
\etkz
&&
\btkz[H,level 1/.style={sd=4},baseline=(0.base)]
\node[inner] (1) at (0,0) {}
 child { node[inner] {} child { node[leaf] {} } }
 child { node[inner] {} child { node[leaf] {} } }
 child { node[leaf] {} }
 child { node[leaf] {} }
 child { node[leaf] {} }
 child { node[leaf] {} };
\node[stage] (0) at (0,.5) {3};
\etkz
&&
\btkz[H,level 1/.style={sd=4.5},baseline=(0.base)]
\node[inner] (1) at (0,0) {}
 child { node[inner] {} child { node[leaf] {} } }
 child { node[leaf] (b) {} }
 child[noedge] { node[noedge] {$\cdots$} }
 child { node[leaf] (a) {} };
\draw [brace]
  (a.north west) -- node [yshift=3ex] {$\scriptstyle 9$} (b.north east);
\node[stage] (0) at (0,.5) {4};
\etkz
&&
\btkz[H,level 1/.style={sd=4.5},baseline=(0.base)]
\node[inner] (1) at (0,0) {}
 child { node[leaf] (b) {} }
 child[noedge] { node[noedge] {$\cdots$} }
 child { node[leaf] (a) {} };
\draw [brace,decoration={raise=1ex}]
  (a.north west) -- node [yshift=3ex] {$\scriptstyle 15$} (b.north east);
\node[stage] (0) at (0,.5) {5};
\etkz
&&
\btkz[H,level 1/.style={sd=4.5},baseline=(0.base)]
\node[inner] (1) at (0,0) {}
 child { node[leaf] (b) {} }
 child[noedge] { node[noedge] {$\cdots$} }
 child { node[leaf] (a) {} };
\draw [brace,decoration={raise=1ex}]
  (a.north west) -- node [yshift=3ex] {$\scriptstyle 14$} (b.north east);
\node[stage] (0) at (0,.5) {6};
\etkz
&&
\cdots
&&
\btkz[H,level 1/.style={sd=4.5},baseline=(0.base)]
\node[leaf] (1) at (0,0) {};
\node[stage] (0) at (0,.5) {20};
\etkz
\end{align*}
By interpreting $\h$, $\I$, and $\si$ as $1$, the power of $\omega$, and
natural addition on ordinals, respectively,
the sequence of Hydras turns into the decreasing sequence of ordinals:
\[
\omega^{\omega^\omega}
~>~ \omega^{\omega^2}
~>~ \omega^{\omega \cdot 3}
~>~ \omega^{\omega \cdot 2 + 4}
~>~ \omega^{\omega + 9}
~>~ \omega^{15}
~>~ \omega^{14}
~>~ \cdots
~>~ 1
\]
One can verify that in general every transition reduces the ordinal
interpretation of the Hydra. Because the order $>$ on ordinals is
well-founded, the termination is concluded.

In the case of the term rewriting encoding, the
mutual dependence between the function symbols $\m{A}$ and $\m{B}$
in rules 3 and 14 of $\xH$ makes proving termination of $\xH/\AC$ a
non-trivial task. We use the technique of semantic labeling 
(Zantema~\cite{Z95}) to resolve the dependence by labeling both
$\m{A}$ and $\m{B}$ by the ordinal value of the Hydra encoded in their
second arguments. Semantic labeling for rewriting modulo has been
investigated in \cite{OMG00}. We need, however, a version for many-sorted
rewriting since the distinction between ordinals and natural numbers is
essential for the effectiveness of semantic labeling for $\xH/\AC$.

Before introducing semantic labeling, we recall some basic
semantic definitions. An \emph{algebra} $\xA$ for an $\xS$-sorted
signature $\xF$
is a pair $(\SET{S_\xA}_{S \in \xS},\SET{f_\xA}_{f \in \xF})$, where
each $S_\xA$ is a non-empty set, called the
\emph{carrier of sort $S$}, and
each $f_\xA$ is a function of type
$f : (S_1)_\xA \times \cdots \times (S_n)_\xA \to S_\xA$, called the
\emph{interpretation function} of
$f : S_1 \times \cdots \times S_n \to S$.
A mapping that associates each variable of sort $S$ to an element in
$S_\xA$ is called an \emph{assignment}. 
We write $\xA^\xV$ for the set of all assignments.
Given an assignment $\alpha \in \xA^\xV$, the \emph{interpretation} of a
term $t$ is inductively defined as follows:
\begin{align*}
\EVALA{t} = \begin{cases}
\alpha(t) &\text{if $t$ is a variable} \\
f_\xA(\EVALA{t_1},\dots,\EVALA{t_n}) &\text{if $t = f(\seq{t})$}
\end{cases}
\end{align*}
Let $\xA = (\SET{S_\xA}_{S \in \xS},\SET{f_\xA}_{f \in \xF})$ be an
$\xS$-sorted $\xF$-algebra. We assume that each carrier set $S_\xA$
is equipped with a well-founded order $>_S$ such that the interpretation
functions are weakly monotone. Here a function 
$\phi$ of type $A_1 \times \cdots \times A_n \to B$ is
\emph{weakly monotone} if $\phi(a_1, \dots, a_i, \dots, a_n) \geqslant_B
\phi(a_1, \dots, b, \dots, a_n)$ whenever $a_i \geqslant_{A_i} b$. We call
$(\xA,\SET{>_S}_{S \in \xS})$ a weakly monotone many-sorted algebra.
Given terms $s$ and $t$ of sort $S$, we write $s \geqslant_\xA t$ 
($s =_\xA t$) if
$[\alpha]_\xA(s) \geqslant_S [\alpha]_\xA(t)$
($[\alpha]_\xA(s) =_S [\alpha]_\xA(t)$)
holds for all $\alpha \in \xA^\xV$.

A labeling $L$ for $\xF$ consists of sets of labels
$L_f \subseteq S_\xA$ for every $f : S_1 \times \cdots \times S_n \to S$.
The labeled signature $\xFl$
consists of function symbols $f_a : S_1 \times \cdots \times S_n \to S$
for every function symbol $f : S_1 \times \cdots \times S_n \to S$ in
$\xF$ and label $a \in L_f$ together with all
function symbols $f \in \xF$ such that $L_f = \varnothing$. A
\emph{labeling} $(L,\m{lab})$ for $(\xA,\SET{>_S}_{S \in \xS})$ consists
of a labeling $L$ for the signature $\xF$ together with a mapping
$\m{lab}_f\colon (S_1)_\xA \times \cdots \times (S_n)_\xA \to L_f$ for
every function symbol $f : S_1 \times \cdots \times S_n \to S$ in
$\xF$ with $L_f \neq \varnothing$.
We call $(L,\m{lab})$ \emph{weakly monotone} if all its labeling
functions $\m{lab}_f$ are weakly monotone.
The mapping $\m{lab}_f$ determines the label of the root symbol $f$ of
a term $f(\seq{t})$, based on the values of its arguments $\seq{t}$.
Formally, for every assignment $\alpha \in \xA^\xV$ we define a
mapping $\m{lab}_\alpha$ inductively as follows:
\begin{gather*}
\lab{t} = \begin{cases}
t &\text{if $t \in \xV$} \\
f(\lab{t_1},\dots,\lab{t_n})
&\text{if $t = f(\seq{t})$ and $L_f = \varnothing$} \\
f_a(\lab{t_1},\dots,\lab{t_n})
&\text{if $t = f(\seq{t})$ and $L_f \neq \varnothing$}
\end{cases}
\end{gather*}
where $a = \m{lab}_f(\EVALA{t_1},\dotsc,\EVALA{t_n})$.
Note that $\lab{t}$ and $t$ have the same sort.
Given a TRS $\xR$ over a (many-sorted) signature $\xF$, we define the
\emph{labeled} TRS $\xRl$ over the signature $\xFl$ as follows:
\begin{gather*}
\xRl = \SET{\lab{\ell} \R \lab{r} \mid
\text{$\ell \R r \in \xR$ and $\alpha \in A^\xV$}}
\end{gather*}
Since there is no need to label the AC symbol $\si$ in the encoding of
the Hydra Battle, we assume
for simplicity that $L_f = \varnothing$ for every AC symbol $f \in \xF$.
The TRS $\DEC$ consists of all rewrite rules
\begin{gather*}
f_a(\seq{x}) \R f_b(\seq{x})
\end{gather*}
with $f : S_1 \times \cdots \times S_n \to S$ a function symbol in $\xF$,
$a, b \in L_f$ such that $a >_S b$, and pairwise different variables
$\seq{x}$.
A weakly monotone algebra $(\xA,>)$ is a \emph{quasi-model} of $\xR/\AC$
if $\ell \geqslant_{\xA} r$ for all rewrite rules $\ell \R r$ in $\xR$ 
and $\ell =_\xA r$ for all equations $\ell \approx r$ in $\AC$.
So in a quasi-model, AC symbols are interpreted as associative
and commutative functions.

\begin{thm}
\label{thm:semantic-labeling}
Let $\xR/\AC$ be a \textup{TRS} over a many-sorted signature $\xF$,
$(\xA,\SET{>_S}_{S \in \xS})$ a quasi-model of $\xR/\AC$ with a
weakly monotone labeling $(L,\m{lab})$. If $(\xRl \cup \DEC) / \AC$
is terminating then $\xR/\AC$ is terminating.
\end{thm}

\begin{proof}
We show
\begin{enumerate}
\Item
if $t \Rb[\xR] u$ then $\lab{t} \Rab[\DEC]{*} \cdot \Rb[\xRl] \lab{u}$
\Item
if $t \Rb[\AC] u$ then $\lab{t} \ac \lab{u}$
\Item
if $t \ac u$ then $\lab{t} \ac \lab{u}$
\end{enumerate}
\smallskip
for all sorts $S$, terms $t, u \in \xTS$, and assignments
$\alpha \in \xA^\xV$. 
The claim follows from the first and third statements.
First suppose $t \Rb[\xR] u$ is a root step
using the rewrite rule $\ell \R r$. So $t = \ell\sigma$ and $u = r\sigma$
for some substitution $\sigma$. Define the assignment
$\beta = [\alpha]_\xA \circ \sigma$ and the (labeled) substitution
$\tau = \m{lab}_\alpha \circ \sigma$. An easy induction proof yields
$\lab{s\sigma} = \lab[\beta]{s}\tau$ for all terms $s$. By definition
$\lab[\beta]{\ell} \R \lab[\beta]{r} \in \xRl$. Hence
$\lab{t} = \lab[\beta]{\ell}\tau \Rb[\xRl] \lab[\beta]{r}\tau = \lab{u}$.
Next suppose $t \Rb[\xR] u$ takes place below the root. So
$t = f(t_1, \dots, t_i, \dots t_n)$ and
$u = f(t_1, \dots, u_i, \dots t_n)$ with $t_i \Rb[\xR] u_i$. Let
$S_1 \times \cdots \times S_n \to S$ be the sort declaration of $f$. The
induction hypothesis yields
$\lab{t_i} \Rab[\DEC]{*} \cdot \Rb[\xRl] \lab{u_i}$. We obtain
$\EVALA{t_i} \geqslant_{S_i} \EVALA{u_i}$ from the quasi-model
assumption. If $L_f = \varnothing$ then
\begin{align*}
\lab{t} = {} &f(\lab{t_1}, \dots, \lab{t_i}, \dots, \lab{t_n})
\Rab[\DEC]{*} \cdot \Rb[\xRl] \\
&f(\lab{t_1}, \dots, \lab{u_i}, \dots, \lab{t_n}) = \lab{u}
\end{align*}
Suppose $L_f \neq \varnothing$ and let
\begin{align*}
a &= \m{lab}_f(\EVALA{t_1}, \dots, \EVALA{t_i}, \dots, \EVALA{t_n}) \\
b &= \m{lab}_f(\EVALA{t_1}, \dots, \EVALA{u_i}, \dots, \EVALA{t_n})
\end{align*}
We obtain $a \geqslant_S b$ from the weak monotonicity of the labeling
function $\m{lab}_f$. 
Therefore, the following rewrite sequence is constructed:
\begin{align*}
\lab{t} = {} &f_a(\lab{t_1}, \dots, \lab{t_i}, \dots, \lab{t_n})
\Rab[\DEC]{*} \\
&f_b(\lab{t_1}, \dots, \lab{t_i}, \dots, \lab{t_n})
\Rab[\DEC]{*} \cdot \Rb[\xRl] \\
&f_b(\lab{t_1}, \dots, \lab{u_i}, \dots, \lab{t_n}) = \lab{u}
\end{align*}
This concludes the proof of the first statement. 
The second statement is shown in the same way, but since
AC symbols are not labeled the rules of $\DEC$ do not come into play.
The third statement is obtained from the
the second statement together with the fact that $t \ac u$ implies
$t \Rab[\AC]{*} u$.
\end{proof}

After these preliminaries, we are ready to put many-sorted semantic
labeling to the test.
Consider the many-sorted algebra $\xA$ with carriers $\NN$ for sort
$\m{N}$ and $\OO$, the set of ordinal numbers smaller than
$\epsilon_0$, for sorts $\m{O}$ and $\m{S}$ and
the following interpretation functions:
\begin{align*} 
\m{0}_\xA &= \h_\xA = 1 & \m{s}_\xA(n) &= n+1 &
\I_\xA(x) &= \omega^x \\
x \si_\xA y &= x \oplus y & \m{E}_\xA(x) &= x+1 &
\m{C}_\xA(n,x) &= x \otimes n + 1 \\
\m{A}_\xA(n,x) &= \ML{$\m{B}_\xA(n,x) = \m{D}_\xA(n,x) = x$}
\end{align*}
Here $\oplus$ denotes natural addition on ordinals and $\otimes$
denotes natural product characterized by
$x \otimes 0 = 0$ and
$x \otimes (n+1) = (x \otimes n) \oplus x$.
Both satisfy strict monotonicity.

\begin{lem}
The algebra $(\xA,\SET{>_\m{O},>_\m{N}})$ is a quasi-model of $\xH/\AC$.
\end{lem}

\begin{proof}
First note that the interpretation functions are weakly monotone.
The rewrite rules in $\xH$ are oriented by $\geO$:
{\allowdisplaybreaks\begin{align*}
\m{A}_\xA(n,\I_\xA(\h_\xA)) = \omega
&\GO 1 = \m{A}_\xA(\m{s}_\xA(n),\h_\xA) 
\tag{1}
\\
\m{A}_\xA(n,\I_\xA(\h_\xA \si_\xA x)) = \omega^{x + 1}
&\GO \omega^x = \m{A}_\xA(\m{s}_\xA(n),\I_\xA(x))
\tag{2}
\\
\m{A}_\xA(n,\I_\xA(x)) = \omega^x
&\EO \omega^x = \m{B}_\xA(n,\m{D}_\xA(\m{s}_\xA(n),\I_\xA(x)))
\tag{3}
\\
\m{C}_\xA(\m{0}_\xA,x) = x + 1
&\EO x + 1 = \m{E}_\xA(x)
\tag{4}
\\
\m{C}_\xA(\m{s}_\xA(n),x) = (x \otimes n) \oplus x + 1
&\EO (x \otimes n) \oplus x + 1 = x \si_\xA \m{C}_\xA(n,x)
\tag{5}
\\ 
\I_\xA(\m{E}_\xA(x) \si_\xA y) = \omega^{x \Oplus y \plus 1}
&\GO \omega^{x \Oplus y} \plus 1 = \m{E}_\xA(\I_\xA(x \si_\xA y))
\tag{6}
\\
\I_\xA(\m{E}_\xA(x)) = \omega^{x+1}
&\GO \omega^x + 1 = \m{E}_\xA(\I_\xA(x))
\tag{7}
\\
\m{D}_\xA(n,\I_\xA(\I_\xA(x))) = \omega^{\omega^x}
&\EO \omega^{\omega^x} = \I_\xA(\m{D}_\xA(n,\I_\xA(x)))
\tag{8}
\\
\m{D}_\xA(n,\I_\xA(\I_\xA(x) \si_\xA y)) = \omega^{\omega^x \Oplus y}
&\EO \omega^{\omega^x \Oplus y} = \I_\xA(\m{D}_\xA(n,\I_\xA(x)) \si_\xA y)
\tag{9}
\\
\makebox[15em][r]{$\m{D}_\xA(n,\I_\xA(\I_\xA(\h_\xA \si_\xA x) \si_\xA y))
= \omega^{\omega^{x+1} \Oplus y}$}
&\GO \omega^{(\omega^x \Otimes n) \Oplus y \plus 1} =
\I_\xA(\m{C}_\xA(n,\I_\xA(x)) \si_\xA y)
\tag{10}
\\
\m{D}_\xA(n,\I_\xA(\I_\xA(\h_\xA \si_\xA x))) = \omega^{\omega^{x+1}}
&\GO \omega^{(\omega^x \Otimes n) \plus 1} =
\I_\xA(\m{C}_\xA(n,\I_\xA(x)))
\tag{11}
\\
\m{D}_\xA(n,\I_\xA(\I_\xA(\h_\xA) \si_\xA y)) = \omega^{\omega \Oplus y}
&\GO \omega^{(n \plus 1) \Oplus y} = \I_\xA(\m{C}_\xA(n,\h_\xA) \si_\xA y)
\tag{12}
\\
\m{D}_\xA(n,\I_\xA(\I_\xA(\h_\xA))) = \omega^\omega
&\GO \omega^{n \plus 1} = \I_\xA(\m{C}_\xA(n,\h_\xA))
\tag{13}
\\
\m{B}_\xA(n,\m{E}_\xA(x)) = x+1
&\GO x = \m{A}_\xA(\m{s}_\xA(n),x)
\tag{14}
\end{align*}}
Note that inequalities (10)---(13) use the fact that
$\omega \gO n$ holds for $n \in \NN$.
The compatibility of $\xA$ with $\AC$ follows from the associativity
and the commutativity of $\oplus$:
\begin{align*}
(x \si_{\xA} y) \si_{\xA} z = (x \oplus y) \oplus z
&\EO x \oplus (y \oplus z) = x \si_{\xA} (y \si_{\xA} z) \\
x \si_{\xA} y = x \oplus y &\EO y \oplus x = x \si_{\xA} y 
\end{align*}
Therefore, $\xA$ is a quasi-model of $\xH/\AC$.
\end{proof}

We now label $\m{A}$ and $\m{B}$ by the value of their second argument.
Let $L_\m{A} = L_\m{B} = \OO$ and $L_f = \varnothing$ for the other
function symbols $f$, and define $\m{lab}$ as follows:
\begin{gather*}
\m{lab}_{\m{A}}(n,x) = \m{lab}_{\m{B}}(n,x) = x
\end{gather*}
The labeling $(L,\m{lab})$
results in the infinite rewrite system $\xHl \cup \DEC$ with
$\xHl$ consisting of the rewrite rules
\begin{align*}
\m{A}_\omega(n,\I(\h)) &\nR{1} \m{A}_1(\m{s}(n),\h) &
\m{D}(n,\I(\I(x))) &\nR{8} \I(\m{D}(n,\I(x))) \\
\m{A}_{\omega^{v+1}}(n,\I(\h \si x)) &\nR{2}
\m{A}_{\omega^v}(\m{s}(n),\I(x)) &
\m{D}(n,\I(\I(x) \si y)) &\nR{9} \I(\m{D}(n,\I(x)) \si y) \\
\m{A}_{\omega^v}(n,\I(x)) &\nR{3}
\m{B}_{\omega^v}(n,\m{D}(\m{s}(n),\I(x))) &
\m{D}(n,\I(\I(\h \si x) \si y)) &\nR{10} \I(\m{C}(n,\I(x)) \si y) \\
\m{C}(\m{0},x) &\nR{4} \m{E}(x) &
\m{D}(n,\I(\I(\h \si x))) &\nR{11} \I(\m{C}(n,\I(x))) \\
\m{C}(\m{s}(n),x) &\nR{5} x \si \m{C}(n,x) &
\m{D}(n,\I(\I(\h) \si y)) &\nR{12} \I(\m{C}(n,\h) \si y) \\
\I(\m{E}(x) \si y) &\nR{6} \m{E}(\I(x \si y)) &
\m{D}(n,\I(\I(\h))) &\nR{13} \I(\m{C}(n,\h)) \\
\I(\m{E}(x)) &\nR{7} \m{E}(\I(x)) &
\m{B}_{v+1}(n,\m{E}(x)) &\nR{14} \m{A}_v(\m{s}(n),x)
\end{align*}
for all $v \in \OO$ and $\DEC$ consisting of the rewrite rules
\begin{align*}
\m{A}_v(n,x) &\,\R\, \m{A}_w(n,x) &
\m{B}_v(n,x) &\,\R\, \m{B}_w(n,x)
\end{align*}
for all $v, w \in \OO$ with $v > w$.

\begin{exa}
The first rewrite sequence in \exaref{fight} is simulated as follows:
{\allowdisplaybreaks\begin{align*}
\m{A}_u(\m{0},\I(\I(\I(\h))))
\nR{3} {} &\m{B}_u(\m{0},\m{D}(\m{s}(\m{0}),\I(\I(\I(\h))))) \\
\nR{8} {} &\m{B}_u(\m{0},\I(\m{D}(\m{s}(\m{0}),\I(\I(\h))))) \\
\nR{13} {} &\m{B}_u(\m{0},\I(\I(\m{C}(\m{s}(\m{0}),\h))))
\,\R_{\DEC}\, \m{B}_{v+1}(\m{0},\I(\I(\m{C}(\m{s}(\m{0}),\h)))) \\
\nR{5} {} &\m{B}_{v+1}(\m{0},\I(\I(\h \si \m{C}(\m{0},\h)))) \\
\nR{4} {} &\m{B}_{v+1}(\m{0},\I(\I(\h \si \m{E}(\h))))
\,\ac\, \m{B}_{v+1}(\m{0},\I(\I(\m{E}(\h) \si \h))) \\
\nR{6} {} &\m{B}_{v+1}(\m{0},\I(\m{E}(\I(\h \si \h)))) \\
\nR{7} {} &\m{B}_{v+1}(\m{0},\m{E}(\I(\I(\h \si \h)))) \\
\nR{14} {} &\m{A}_v(\m{s}(\m{0}),\I(\I(\h \si \h)))
\end{align*}}%
Here $u = \omega^{\omega^\omega}$ and $v = \omega^{\omega^2}$.
\end{exa}

According to \thmref{semantic-labeling}, the AC termination of $\xH$
on many-sorted terms follows from the AC termination of $\xHl \cup \DEC$.

\begin{cor}
\label{cor:termination H labeled => H}
If $\xHl \cup \DEC$ is \textup{AC} terminating,
$\xH$ is \textup{AC} terminating on sorted terms.
\qed
\end{cor}

\section{AC-MPO}
\label{sec:acmpo}

In order to show AC termination of $\xHl \cup \DEC$ we use a simplified
version of AC-RPO.

\begin{defi}
Let $\xF_\AC$ be the set of AC symbols in $\xF$.
Given a non-variable term $t = f(\seq{t})$, the multiset $\TF{t}$ 
is defined inductively as follows:
\begin{align*}
\TF{t} &= \TF[f]{t_1} \uplus \cdots \uplus \TF[f]{t_n} \\
\TF[f]{t} &= \begin{cases}
\TF[f]{t_1} \uplus \TF[f]{t_2} &\text{if $t = f(t_1,t_2)$ and
$f \in \xF_\AC$} \\
\SET{t} &\text{otherwise}
\end{cases}
\end{align*}
\end{defi}

For example, if $+$ is an AC symbol, we have
$\TF[+]{\m{a} + (\m{b} + x)} = \SET{\m{a},\m{b},x}$. If $f$ is a non-AC
symbol, we have $\TF{f(\seq{t})} = \SET{\seq{t}}$.

The multiset extension $\ACm$ of the equivalence relation $=_\AC$ is
inductively defined as follows: $\varnothing \ACm \varnothing$ and
$\SET{s} \uplus M \ACm \SET{t} \uplus N$ if $s =_\AC t$ and 
$M \ACm N$.
It is not difficult to see that $\ACm$ is an equivalence relation.
We have $\TF{s} \ACm \TF{t}$ whenever $s =_\AC t$.

\begin{defi}
\label{def:acmpo}
\emph{Precedences} are strict orders on function symbols.
Let $>$ be a precedence. We define $\ACMPO$ inductively as follows:
$s \ACMPO t$ if $s \notin \xV$
and one of the following conditions holds:
\begin{enumerate}
\Item
$\TF{s} \ACMPOm[\geqslant] \SET{t}$,
\Item
$\RT{s} > \RT{t}$ and $\SET{s} \ACMPOm \TF{t}$,
\Item
$\RT{s} = \RT{t}$ and $\TF{s} \ACMPOm \TF{t}$.
\smallskip
\end{enumerate}
In the third condition $=_\AC$ is used instead of $=$ in the
definition of multiset extension.
We write $\ACMPO[\geqslant]$ for the union of $\ACMPO$ and
$=_\AC$.
\end{defi}

The first condition is equivalent to $s' \ACMPO[\geqslant] t$ for some
$s' \in \TF{s}$ and the second condition is equivalent to
the conjunction of $\RT{s} > \RT{t}$ and
$s \ACMPO t'$ for all $t' \in \TF{t}$. These equivalences will be used
freely in the sequel.
The multiset comparison in the third condition is spelled out as
follows:
$\TF{s} \ACMPOm \TF{t}$ if there exist multisets $S_1$, $S_2$, $T_1$ and
$T_2$ such that $\TF{s} = S_1 \uplus S_2$,
$\TF{t} = T_1 \uplus T_2$, $S_1 \ACm T_1$, $S_2 \neq \varnothing$,
and for every $t' \in T_2$ there exists a term $s' \in S_2$ such that
$s' \ACMPO t'$.

Note that if there are no AC symbols, the above definition reduces to the
original recursive path order of Dershowitz~\cite{D82}, nowadays known
as the \emph{multiset path order}.
Moreover, if AC symbols are minimal
in a precedence, AC-RPO with the multiset status reduces to AC-MPO.
Hence the simplified AC-RPO will be called AC-MPO.  

The proof of the following result can be found in the appendix.
\emph{Incrementality} of AC-MPO 
means that for precedences $>$ and $\sqsupset$ the inclusion 
${\ACMPO} \subseteq {\ACMPO[\sqsupset]}$ holds
whenever ${>} \subseteq {\sqsupset}$.

\begin{thm}
\label{thm:acmpo}
If \textup{AC} symbols are minimal in the precedence $>$ then
$\ACMPO$ is an incremental \textup{AC}-compatible rewrite order with
the subterm property.
\qed
\end{thm}

As a consequence, $\ACMPO$ is an AC-compatible reduction order when
the underlying signature is finite. This also holds for
infinite signatures, provided the precedence $>$
is well-founded and there are only finitely many AC symbols.
This extension is important because the signature of $\xHl$ is
infinite. Below, we will formally prove the correctness of the extension,
by adopting the approach of \cite{MZ97}.

A strict order $>$ on a set $A$ is a \emph{partial well-order} if for
every infinite sequence $a_0, a_1, \dots$ of elements in $A$ there
exist indices $i$ and $j$ such that $i < j$ and $a_i \leqslant a_j$.
Well-founded total orders (\emph{well-orders}) are partial well-orders.
Given a partial well-order $>$ on $\xF$, the \emph{embedding} TRS
$\Emb(\xF,>)$ consists of the rules $f(\seq{x}) \to x_i$ for
every $n$-ary function symbol and $1 \leqslant i \leqslant n$, together
with the rules $f(\seq{x}) \to g(x_{i_1},\dots,x_{i_m})$
for all function symbols $f$ and $g$ with arities $m$ and $n$ such
that $f > g$, and indices
$1 \leqslant i_1 < i_2 < \cdots < i_m \leqslant n$.
Here $\seq{x}$ are pairwise distinct variables.

\begin{thmC}[{\cite[Theorem~5.3]{MZ97}}]
\label{thm:MZ97}
A rewrite order $>$ is well-founded if 
$\Emb(\xF,{\sqsupset}) \subseteq {>}$ for some partial well-order
$\sqsupset$.
\qed
\end{thmC}

\begin{thm}
\label{thm:acmpo infinite}
Consider a signature $\xF$ with only finitely many \textup{AC} symbols
that are minimal in a given well-founded precedence $>$. The relation
$\ACMPO$ is an \textup{AC}-compatible reduction order.
\end{thm}

\begin{proof}
We only need to show 
well-foundedness of $>_\acmpo$ because the other properties follow by
\thmref{acmpo}.
Let $\sqsupset$ be an arbitrary partial well-order that
contains $>$ and in which AC symbols are minimal.
The
inclusion $\Emb(\xF,{\sqsupset}) \subseteq {\sqsupset_\acmpo}$
is easily verified. Hence the well-foundedness of $\sqsupset_\acmpo$
is obtained from \thmref{MZ97}. Since ${>} \subseteq {\sqsupset}$, the
incrementality of AC-MPO yields
${>_\acmpo} \subseteq {\sqsupset_\acmpo}$. It follows that 
$>_\acmpo$ is well-founded.
\end{proof}

We show the termination of $\xHl \cup \DEC$ by AC-MPO. To this end,
we consider the following precedence $>$ on the labeled signature:
\begin{alignat*}{2}
\m{A}_v &> \m{A}_w &\quad&
\text{for all $v, w \in \OO$ with $v > w$} \\
\m{B}_v &> \m{B}_w &\quad&
\text{for all $v, w \in \OO$ with $v > w$} 
\\
\m{B}_{v+1} &> \m{A}_v > \m{B}_v &&
\text{for all $v \in \OO$} \\
\m{B}_0 &> \ML{$\m{s} > \m{D} > \m{C} > \I > \m{E} > \si$}
\end{alignat*}
Note that $>$ is well-founded and the only AC symbol $\si$ is minimal.
In order to ease the compatibility
verification we employ the following simple criterion.

\begin{lem}
\label{lem:root}
Let $\ell \to r$ be a rewrite rule and let $>$ be a precedence.
If $\RT{\ell} > g$ for all function symbols $g$ in $r$ then
$\ell \ACMPO r$.
\qed
\end{lem}

\begin{thm}
\label{thm:termination of H labeled}
$\xHl \cup \DEC \,\subseteq\, {>_\acmpo}$
\end{thm}

\begin{proof}
\lemref{root} applies to all rules of $\xHl \cup \DEC$, except 5\,--\,9.
We consider rule 6 here; the other rewrite rules are handled in a similar
fashion. 
Since case (1) of \defref{acmpo} yields $\m{E}(x) \ACMPO x$, we have
$\TF{\m{E}(x) \si y} = \SET{\m{E}(x),y} \ACMPOm \SET{x,y} = \TF{x \si y}$.
Thus $\m{E}(x) \si y \ACMPO x \si y$ follows by case (3).
Using case (3) again, we obtain $\I(\m{E}(x) \si y) \ACMPO \I(x \si y)$.
Because of $\I > \m{E}$, the desired orientation $\I(\m{E}(x) \si y)
\ACMPO \I(x \si y)$ is concluded by case (2).
\end{proof}

\begin{thm}
\label{thm:termination H labeled}
The \textup{TRS} $\xHl \cup \DEC$ is \textup{AC} terminating.
\qed
\end{thm}

From Theorems~\ref{thm:simulation} and~\ref{thm:termination of H labeled}
we conclude that Hercules eventually beats Hydra in any battle.
Theorems~\ref{thm:termination of H labeled} and \ref{thm:persistency} in
connection with \corref{termination H labeled => H} yield the
AC termination of $\xH$ on arbitrary terms.

\section{Related Work}
\label{sec:related}

In an influential survey paper, Dershowitz and
Jouannaud~\cite[p.~270]{DJ90}
introduced a 5-rule rewrite system to simulate the Hydra Battle. The
proposed rewrite system was later shown to be erroneous. A corrected
version together with a detailed termination analysis has been given by
Dershowitz and Moser~\cite{DM07}, see also Moser~\cite{M09}. Earlier,
Touzet~\cite{T98} presented an 11-rule rewrite system
that encodes a specific battle with weakened
Hydras (whose height is bounded by 4) and proved total termination
by a semantic termination method.
It is worth noting that our rewrite system $\xH$ is not even simply
terminating on unsorted terms. In fact, we have the following cyclic
sequence with respect to $\xH \cup \Emb(\xF,\varnothing)$:
\begin{align*}
\m{A}(\m{E}(\I(x)),\I(x))
&\nR{3} \m{B}(\m{E}(\I(x)),\m{D}(\m{s}(\m{E}(\I(x))),\I(x)))
\Rab[\Emb(\xF,\varnothing)]{*} \m{B}(\m{E}(\I(x)),\I(x)) \\
&\nR{14} \m{A}(\m{s}(\m{E}(\I(x))),\I(x))
\Rb[\Emb(\xF,\varnothing)] \m{A}(\m{E}(\I(x)),\I(x)) 
\end{align*}
So the TRS $\xH$ is not simply terminating (see \cite[Lemma~4.6]{MZ97}).
This is the reason that our termination proof employs semantic
labeling.

The rewrite systems referred to above model the so-called
\emph{standard} battle, which corresponds to a specific strategy for
Hercules. In this regard it is interesting to quote Kirby and
Paris~\cite{KP82}, who introduced the battle as an accessible example of
an independence result for Peano arithmetic (P):
\begin{quote}
A \emph{strategy} is a function which determines for Hercules which head
to chop off at each stage of any battle. It is not hard to find a
reasonably fast \emph{winning strategy} (i.e.\ a strategy which ensures
that Hercules wins against any hydra). More surprisingly, Hercules cannot
help winning: \medskip \\
Theorem 2.\,(i) ~ \emph{Every strategy is a winning strategy.} \medskip \\
\null \qquad $[\dots]$ \medskip \\
Theorem 2.\,(ii) ~ \emph{The statement ``every recursive strategy is a
winning strategy'' is not provable from} P.
\end{quote}
In a recent paper~\cite[Section~6]{EKO21}, rules are
presented to slay Hydras, independent of the strategy. These rules
do not constitute a term rewrite system in the usual sense
(they operate on terms with \emph{sequence variables}). More importantly,
the infinitely many rules do not faithfully represent the battle.
Earlier, Ferreira and Zantema~\cite[Section~10]{FZ96} presented an
infinite rewrite system to model the standard strategy and gave a
direct ordinal interpretation to conclude its termination.
In neither of the latter two papers stages of the battle are modeled.

\section{Conclusion}
\label{sec:conclusion}

We presented a new TRS encoding of the Battle of Hydra and Hercules.
Unlike earlier encodings, it makes use of AC symbols. This allows us
to faithfully model any strategy of Hercules, as envisaged in the
paper by Kirby and Paris~\cite{KP82} in which the Battle was first
presented. To prove the termination of the encoding we
employed many-sorted rewriting modulo AC and we extended semantic
labeling modulo AC to many-sorted TRSs. The infinite TRS produced
by semantic labeling was proved terminating by suitably instantiating
AC-RPO.

One of the reviewers for this article pointed out that the Hydra
battle can still be simulated even if rule $6$ in $\xH$ is replaced
by the simpler $\m{E}(x) \si y \R \m{E}(x \si y)$. While we expect
that this variant also has the termination property, the presented
termination methods are not applicable. In fact, the rule cannot be
ordered by AC-MPO and the AC symbol $\si$ cannot be labeled.
The reviewer also suggested an alternative encoding of Hydras that
omits $\I$ from $\I(t_1 \si \cdots \si t_n)$. For instance, $H_0$ in
\exaref{hydras} is written as $\I(\h) \si \I(\I(\h \si \h)) \si \h$ in
this encoding. While this simplifies representations of Hydras, it seems
difficult to construct an ordinal interpretation of $\si$ for semantic
labeling. Further investigations of AC termination techniques are
required.

The finite TRS $\xH$ poses an interesting challenge for automatic
termination tools. None of the tools (\AProVE~\cite{APROVE},
\muterm~\cite{MUTERM}) competing in the ``TRS Equational'' category of 
the Termination Competition
2024\footnote{https://termcomp.github.io/Y2024/} succeeds on $\xH/\AC$.
This is not really surprising since most methods
implemented in termination tool come with a multiple recursive
upper bound on the derivation height~(e.g.\ \cite{H92,L01,MS11}).
The tools even fail to prove termination of $\xH$ without AC.
The tool \TTTT~\cite{TTT2} has support for ordinal
interpretations~\cite{ZWM15} but also fails on $\xH$.

Formalizing the techniques used in this article in a proof assistant is an
important task to ensure the correctness of the results. Interestingly,
the informal paper~\cite{HM22} in which we announced our encoding also
presents a termination proof, essentially extending a semantic
method of Touzet~\cite{T98} and Zantema~\cite{Z01} to AC rewriting.
Although we believe the non-trivial extension to be correct, its use in
proving the AC termination of $\xH$ has a critical mistake,
which we recently discovered.

Another topic for future research is to investigate the scope of
many-sorted semantic labeling. Can the termination of earlier encodings
of the battle be established with many-sorted semantic labeling followed
by some standard simplification order? Variants of the battle
by Buchholz~\cite{B87} and Lepper~\cite{L04} are also of interest here.

\section*{Acknowledgements}

We are grateful to the anonymous reviewers for their pertinent comments,
which helped to improve the presentation. We thank Teppei Saito for his
thorough feedback on the proofs for AC-MPO.

\bibliographystyle{alphaurl}
\bibliography{references}

\appendix

\section{Proof of \thmref{acmpo}}

We first show the AC-compatibility of AC-MPO.

\begin{lem}
\label{lem:AC compatibility}
The relation $\ACMPO$ is \textup{AC}-compatible.
\end{lem}

\begin{proof}
First assume $s \ACMPO t =_\AC u$. By induction on $|s| + |t|$ we show
$s \ACMPO u$. We distinguish three cases, according to \defref{acmpo}.
\begin{enumerate}
\item
If $s' \ACMPO[\geqslant] t$ for some $s' \in \TF{s}$ then also
$s' \ACMPO[\geqslant] u$, either by the induction hypothesis or by the
transitivity of $=_\AC$. Therefore $s \ACMPO u$ by case (1).
\item
Suppose $\RT{s} > \RT{t}$ and $s \ACMPO t'$ for all $t' \in \TF{t}$.
From $t =_\AC u$ we derive $\RT{t} = \RT{u}$ and $\TF{t} \ACm \TF{u}$. 
So for every $u' \in \TF{u}$ there exists a term $t' \in \TF{t}$ with
$t' =_\AC u'$. Because $s \ACMPO t' =_\AC u'$ and $|t| > |t'|$, the
induction hypothesis yields $s \ACMPO u'$. Hence
$\SET{s} \ACMPOm \TF{u}$ and thus $s \ACMPO u$ by case (2).
\item
Suppose $\RT{s} = \RT{t}$ and $\TF{s} \ACMPOm \TF{t}$. From $t =_\AC u$ we
derive $\RT{t} = \RT{u}$ and $\TF{t} \ACm \TF{u}$. As
$\TF{s} \ACMPOm \TF{t}$,
there exist multisets $S_1$, $S_2$, $T_1$, and $T_2$ such that
$\TF{s} = S_1 \uplus S_2$, $\TF{t} = T_1 \uplus T_2$, $S_1 \ACm T_1$,
$S_2 \neq \varnothing$, and for every $t' \in T_2$ there exists a term
$s' \in S_2$ with $s' \ACMPO t'$. As $\TF{t} \ACm \TF{u}$, we may write
$\TF{u} = U_1 \uplus U_2$ with $T_1 \ACm U_1$ and $T_2 \ACm U_2$. As
$S_1 \ACm T_1 \ACm U_1$, we obtain $S_1 \ACm U_1$ from the transitivity of
$=_\AC$. For every $u' \in U_2$ there exists a term $t' \in T_2$ with
$t' =_\AC u'$. Moreover, there exists a term $s' \in S_2$ with
$s' \ACMPO t'$. Since $s' \ACMPO t' =_\AC u'$ and
$|s| + |t| > |s'| + |t'|$, the induction hypothesis yields $s' \ACMPO u'$.
Consequently, $\TF{s} \ACMPOm \TF{u}$. Hence $s \ACMPO u$ by case (3).
\end{enumerate}
Next assume $s =_\AC t \ACMPO u$. By induction on $|t| + |u|$ we show
$s \ACMPO u$. From $s =_\AC t$ we infer
$\RT{s} = \RT{t}$ and $\TF{s} \ACm \TF{t}$.
We distinguish three cases for $t \ACMPO u$.
\begin{enumerate}
\item
Suppose $\TF{t} \ACMPO[\geqslant] \SET{u}$. Since $\TF{s} \ACm \TF{t}$, we
obtain $\TF{s} \ACMPOm \SET{u}$ by the induction hypothesis or the
transitivity of $=_\AC$. Hence $s \ACMPO u$ by case (1).
\item
Suppose $\RT{t} > \RT{u}$ and $t \ACMPO u'$ for all $u' \in \TF{u}$. The
induction hypothesis yields $s \ACMPO u'$ for all $u' \in \TF{u}$. Since
also $\RT{s} > \RT{u}$, $s \ACMPO u$ by case (2).
\item
Suppose
$\RT{t} = \RT{u}$ and $\TF{t} \ACMPOm \TF{u}$. Since $\TF{s} \ACm \TF{t}$,
we obtain $\TF{s} \ACMPOm \TF{u}$ by the induction hypothesis and the
transitivity of $=_\AC$. Hence $s \ACMPO u$ by case (3).
\qedhere
\end{enumerate}
\end{proof}

Next we show transitivity.

\pagebreak[5]
\begin{lem}
\label{lem:transitivity}
The relation $\ACMPO$ is transitive.
\end{lem}

\begin{proof}
Suppose $s \ACMPO t \ACMPO u$. We show $s \ACMPO u$ by induction on
$|s| + |t| + |u|$. We do a case analysis on $s \ACMPO t$.
\begin{enumerate}
\item
If $s' \ACMPO[\geqslant] t$ for some $s' \in \TF{s}$ then $s' \ACMPO u$ by
the induction hypothesis or the AC-compatibility of $\ACMPO$
(\lemref{AC compatibility}).
\item
Suppose $\RT{s} > \RT{t}$ and $\SET{s} \ACMPOm \TF{t}$. We perform a
second case analysis on $t \ACMPO u$.
\begin{itemize}
\item
If $\TF{t} \ACMPO[\geqslant] \SET{u}$ then we obtain $s \ACMPO u$ by the
induction hypothesis or the AC-compatibility of $\ACMPO$.
\item
If $\RT{t} > \RT{u}$ and $\SET{t} \ACMPO \TF{u}$ then $s \ACMPO t \ACMPO v$
for all $v \in \TF{u}$ and thus $\SET{s} \ACMPOm \TF{u}$ by the
induction hypothesis. Hence $s \ACMPO u$ by case (2).
\item
Suppose $\RT{t} = \RT{u}$ and $\TF{t} \ACMPOm \TF{u}$.
We obtain $\TF{s} \ACMPOm \TF{u}$ from the induction hypothesis and
the AC-compatibility of $\ACMPO$. Thus, $s \ACMPO u$ follows by case (3).
\end{itemize}
\item
Suppose $\RT{s} = \RT{t}$ and $\TF{s} \ACMPOm \TF{t}$. Also in this
case we perform an additional case analysis on $t \ACMPO u$.
\begin{itemize}
\item
If $\TF{t} \ACMPO[\geqslant] \SET{u}$ then we obtain
$\TF{s} \ACMPO \SET{u}$ by the induction hypothesis or the AC-compatibility
of $\ACMPO$.
\item
Suppose $\RT{t} > \RT{u}$ and $\SET{t} \ACMPOm \TF{u}$. We have
$\RT{s} > \RT{u}$. For every $v \in \TF{u}$ we have
$s \ACMPO t \ACMPO v$, and thus $s \ACMPO v$ by the induction hypothesis.
Hence $\SET{s} \ACMPOm \TF{u}$ and thus $s \ACMPO u$ by case (2). 
\item
Suppose $\RT{t} = \RT{u}$ and $\TF{t} \ACMPOm \TF{u}$. From
$\TF{s} \ACMPOm \TF{t} \ACMPOm \TF{u}$ we infer $\TF{s} \ACMPOm \TF{u}$
by the induction hypothesis, the AC-compatibility of $\ACMPO$, and the
transitivity of $=_\AC$. Hence $s \ACMPO u$ by case (3).
\qedhere
\end{itemize}
\end{enumerate}
\end{proof}

The subterm property is next.

\begin{lem}
\label{lem:subterm property}
The relation $\ACMPO$ has the subterm property.
\end{lem}

\begin{proof}
Let $t = f(\seq{t})$. Fix $i \in \SET{1,\dots,n}$. We show
$t \ACMPO t_i$. The subterm property is then obtained by induction and
the transitivity of $\ACMPO$ (\lemref{transitivity}).
We distinguish two cases.
\begin{enumerate}
\item
If $t_i \in \TF{t}$ then $\TF{t} \ACMPO[\geqslant] \SET{t_i}$
and thus $t \ACMPO t_i$ by case (1).
\item
If $t_i \notin \TF{t}$ then $f$ is an AC symbol, $n = 2$ and
$\RT{t_i} = f$. Since
$\TF{t_i} \subsetneq \TF{t}$, $\TF{t} \ACMPOm \TF{t_i}$ holds and
thus $t \ACMPO t_i$ by case (3).
\qedhere
\end{enumerate}
\end{proof}

The preceding lemmata are used to prove irreflexivity.

\begin{lem}
\label{lem:irreflexivity}
The relation $\ACMPO$ is irreflexive.
\end{lem}

\begin{proof}
Assume to the contrary $t \ACMPO t$. We derive a
contradiction by induction on $t$. We distinguish three cases.
\begin{enumerate}
\item
Suppose $t' \ACMPO[\geqslant] t$ for some $t' \in \TF{t}$. Since
$t \prsuperterm t'$, the subterm property (\lemref{subterm property})
yields $t \ACMPO t$. So $t' \ACMPO[\geqslant] t \ACMPO t'$, and thus
$t' \ACMPO t'$ is obtained by the transitivity (\lemref{transitivity}) or
AC compatibility (\lemref{AC compatibility}) of $\ACMPO$.
The induction hypothesis yields the desired contraction.
\item
If $t \ACMPO t$ is derived by case (2) then $\RT{t} > \RT{t}$, which
contradicts the irreflexivity of the precedence $>$.
\item
Suppose $\TF{t} \ACMPOm \TF{t}$. Let $U$ be the set of all 
proper subterms of
$t$, and $\succ$ the restriction of $\ACMPO$ to $U \times U$. The
multiset extension $\succ^\mul$ coincides with the restriction of $\ACMPOm$
to finite multisets over $U$. Hence $\TF{t} \succ^\mul \TF{t}$ follows from
$\TF{t} \ACMPOm \TF{t}$. The relation $\succ$ is irreflexive according to
the induction hypothesis. Moreover, $\succ$ inherits transitivity from
$\ACMPO$ (\lemref{transitivity}). Hence $\succ$ is a strict order and
thus so is its multiset extension $\succ^\mul$. Since $\TF{t}$ is a finite
multiset over $U$, $\TF{t} \succ^\mul \TF{t}$ cannot hold, yielding the
desired contradiction.
\qedhere
\end{enumerate}
\end{proof}

In the proof of closure under substitutions we use the fact that
for an $f$-rooted term $t$ and a substitution $\sigma$ the multiset
$\TF{t\sigma}$ is the multiset sum of $\TF[f]{t'\sigma}$ for all 
$t' \in \TF{t}$.

\begin{lem}
The relation $\ACMPO$ is closed under substitutions.
\end{lem}

\begin{proof}
Suppose $s \ACMPO t$ and let $\sigma$ be a substitution.
We show $s\sigma \ACMPO t\sigma$ by induction on $|s| + |t|$.
We distinguish three cases.
\begin{enumerate}
\item
Suppose $s' \ACMPO[\geqslant] t$ for some $s' \in \TF{s}$. If $s' \ACMPO t$
then we obtain $s'\sigma \ACMPO t\sigma$ from the induction hypothesis.
If $s' =_\AC t$ then we obtain
$s'\sigma =_\AC t\sigma$ from the closure under substitutions of $=_\AC$.
So in both cases we have $s'\sigma \ACMPO[\geqslant] t\sigma$.
If $s'\sigma \in \TF{s\sigma}$ then
$\TF{s\sigma} \ACMPOm \SET{t\sigma}$
and thus $s\sigma \ACMPO t\sigma$ by case (1).
If $s'\sigma \notin \TF{s\sigma}$ then $s' \in \xV$ and
thus $s' = t$ follows from $s' \ACMPO[\geqslant] t$.
Hence $s'\sigma = t\sigma$.
Since $s'\sigma \prsubterm s\sigma$ we obtain
$s\sigma \ACMPO t\sigma$ from the subterm property
(\lemref{subterm property}).
\item
Suppose $\RT{s} > \RT{t}$ and $\SET{s} \ACMPOm \TF{t}$. Clearly
$\RT{s\sigma} = \RT{s} > \RT{t} = \RT{t\sigma}$. Consider an arbitrary
term $u \in \TF{t\sigma}$. Let $f = \RT{t}$. There exists a term
$t' \in \TF{t}$ such that $u \in \TF[f]{t'\sigma}$. The
induction hypothesis yields $s\sigma \ACMPO t'\sigma$.
Since
$u \in \TF[f]{t'\sigma}$ satisfies $u \subterm t'\sigma$,
we have
$t'\sigma \ACMPO[\geqslant] u$ by the subterm property
(\lemref{subterm property})
and thus $s\sigma \ACMPO u$ by transitivity (\lemref{transitivity}).
Hence $\SET{s\sigma} \ACMPOm \TF{t\sigma}$ and thus
$s\sigma \ACMPO t\sigma$ by case (2).
\item
Suppose $\RT{s} = \RT{t}$ and $\TF{s} \ACMPOm \TF{t}$.
We write $f$ for $\RT{s}$.
Let $s' \in \TF{s}$ and $t' \in \TF{t}$ such that
$s' \ACMPO t'$ or $s' =_\AC t'$ is used in $\TF{s} \ACMPOm \TF{t}$.
\begin{itemize}
\Item
We first consider $s' \ACMPO t'$.
Since $s' \prsubterm s$ and $t' \prsubterm t$,
the induction hypothesis yields $s'\sigma \ACMPO t'\sigma$.
From $s' \ACMPO t'$ we infer $s' \notin \xV$ and thus
$s'\sigma \notin \xV$.  Because of $s' \in \TF{s}$, we have
$\RT{s'} \neq f$ and thus $\RT{s'\sigma} \neq f$. Hence
$\TF[f]{s'\sigma} = \SET{s'\sigma}$. Consider a term 
$u \in \TF[f]{t'\sigma}$. Since $u \subterm t'\sigma$, we
obtain $t'\sigma \ACMPO[\geqslant] u$ by the subterm property
(\lemref{subterm property}) and thus $s'\sigma \ACMPO u$ by
transitivity
(\lemref{transitivity}). Hence
$\TF[f]{s'\sigma} \ACMPOm \TF[f]{t'\sigma}$ follows.
\Item
Suppose $s' =_\AC t'$. Since $=_\AC$ is closed under substitutions, we
have $s'\sigma =_\AC t'\sigma$ and thus
$\TF[f]{s'\sigma} \ACm \TF[f]{t'\sigma}$.
\smallskip
\end{itemize}
It follows that the derivation of $\TF{s} \ACMPOm \TF{t}$ 
can be simulated, resulting in $\TF{s\sigma} \ACMPOm \TF{t\sigma}$.
Hence $s\sigma \ACMPO t\sigma$ by case (3).
\qedhere
\end{enumerate}
\end{proof}

The following technical result is used in the proof that AC-MPO is
closed under contexts (if AC symbols are minimal in the precedence).

\begin{lem}
\label{lem:tf}
If $f \in \xF_\AC$ is minimal in $>$ and $s = f(s_1,s_2) \ACMPO t$ then
$\TF[f]{s} \ACMPOm \TF[f]{t}$.
\end{lem}

\begin{proof}
We have $\TF{s} = \TF[f]{s} = \TF[f]{s_1} \uplus \TF[f]{s_2}$.
We distinguish three cases.
\begin{enumerate}
\item
Suppose $s' \ACMPO[\geqslant] t$ for some $s' \in \TF{s}$.
We obtain $\RT{s'} \neq f$ from $s' \in \TF{s} = \TF[f]{s}$
and thus $\TF[f]{s'} = \SET{s'}$.
If also $\RT{t} \neq f$, then $\TF[f]{t} = \SET{t}$ and thus
$s' \ACMPO[\geqslant] t$ leads to 
$\TF[f]{s} \supsetneq \TF[f]{s'} \ACMPOm[\geqslant] \TF[f]{t}$
and hence $\TF[f]{s} \ACMPOm \TF[f]{t}$. Suppose $\RT{t} = f$.
Let $v \in \TF[f]{t}$. We have $t \prsuperterm v$ and thus
$t \ACMPO v$. As $s' \ACMPO[\geqslant] t \ACMPO v$, we obtain
$s' \ACMPO v$ by transitivity or AC-compatibility.
Hence $\TF[f]{s'} \ACMPOm \TF[f]{t}$ and therefore
$\TF[f]{s} \supsetneq \TF[f]{s'} \ACMPOm \TF[f]{t}$ and
$\TF[f]{s} \ACMPOm \TF[f]{t}$.
\item
Since $f$ is minimal in the precedence, $s \ACMPO t$ cannot be obtained
by case (2).
\item
Suppose $\RT{s} = \RT{t} = f$ and $\TF{s} \ACMPOm \TF{t}$.
Since $\TF[f]{s} = \TF{s}$ and $\TF[f]{t} = \TF{t}$, the claim holds.
\qedhere
\end{enumerate}
\end{proof}

\begin{lem}
The relation $\ACMPO$ is closed under contexts if \textup{AC} symbols are
minimal in the precedence $>$.
\end{lem}

\begin{proof}
Suppose $s \ACMPO t$ and consider a 
context of the form $C = f(\dots,\Box,\dots)$.
If $f \notin \xF_\AC$ then
$\TF{C[s]} - \TF{C[t]} = \SET{s}$ and
$\TF{C[t]} - \TF{C[s]} = \SET{t}$ by the irreflexivity of $\ACMPO$
(\lemref{irreflexivity}), and thus
$\TF{C[s]} \ACMPOm \TF{C[t]}$ follows from
$s \ACMPO t$. Hence $C[s] \ACMPO C[t]$ by case (3).
Suppose $f \in \xF_\AC$. We have
$C = f(\Box,u)$ or $C = f(u,\Box)$ and distinguish two cases.
\begin{itemize}
\item
If $f = \RT{s}$ then $\TF[f]{s} \ACMPOm \TF[f]{t}$
by \lemref{tf}. So the inequality
\[
\TF{C[s]} - \TF{C[t]} = \TF[f]{s} \ACMPOm \TF[f]{t} = \TF{C[t]} - \TF{C[s]}
\]
holds. Therefore, we obtain $C[s] \ACMPO C[t]$ by case (3).
\item
If $f \neq \RT{s}$ then $\TF{C[s]} = \SET{s} \uplus \TF[f]{u}$ and
$\TF{C[t]} = \TF[f]{t} \uplus \TF[f]{u}$.
According to case (3), it is enough to show
$s \ACMPO t'$ for all $t' \in \TF[f]{t}$.
Let $t' \in \TF[f]{t}$. We have $t \superterm t'$ and thus
$t \ACMPO[\geqslant] t'$ by the subterm property
(\lemref{subterm property}). As $s \ACMPO t \ACMPO[\geqslant] t'$, we
obtain $s \ACMPO t'$ by transitivity or AC-compatibility.
\qedhere
\end{itemize}
\end{proof}

Incrementality is the final property in \thmref{acmpo}.

\begin{lem}
The relation $\ACMPO$ is incremental.
\end{lem}

\begin{proof}
Let $>$ and $\sqsupset$ be precedences with ${>} \subseteq {\sqsupset}$.
Suppose $s \ACMPO t$. We show $s \ACMPO[\sqsupset] t$ by induction on 
$|s| + |t|$. We distinguish three cases.
\begin{enumerate}
\item
If $\TF{s} \ACMPOm[\geqslant] \SET{t}$ then $s' =_\AC t$ or $s' \ACMPO t$
for some $s' \in \TF{s}$. In the latter case, the induction hypothesis
yields $s' \ACMPO[\sqsupset] t$. Hence in both cases
$s \ACMPO[\sqsupset] t$ by case (1).
\item
If $\RT{s} > \RT{t}$ and $\SET{s} \ACMPOm \TF{t}$ then
$\SET{s} \ACMPO[\sqsupset] \TF{t}$ by the induction hypothesis. Hence,
$s \ACMPO[\sqsupset] t$ is obtained by case (2).
\item
Suppose $\RT{s} = \RT{t}$ and $\TF{s} \ACMPOm \TF{t}$. Let
$s' \in \TF{s}$ and $t' \in \TF{s}$ be a
term pair
such that $s' \ACMPO t'$ is used in $\TF{s} \ACMPOm \TF{t}$. Since 
$s' \prsubterm s$ and $t' \prsubterm t$, the induction hypothesis yields
$s' \ACMPO[\sqsupset] t'$. Thus, the derivation of
$\TF{s} \ACMPOm \TF{t}$ can be simulated by $\ACMPO[\sqsupset]$.
Therefore, $\TF{s} \ACMPOm[\sqsupset] \TF{t}$, and hence 
$s \ACMPO[\sqsupset] t$ is obtained by case (3).
\qedhere
\end{enumerate}
\end{proof}

\end{document}